\newif\ifnotes
\notestrue

\documentclass[12pt]{article}
\ifnotes %
    \usepackage[draft=false, notes=true, marginpar=false, done=false, later=false]{dtrt}
\else
    \usepackage[draft=false, notes=xxx, marginpar=false, done=false, later=false]{dtrt}
\fi

\usepackage[top=1in, bottom=1in, left=1in, right=1in]{geometry}

\usepackage{natbib}

\let\cite\citeyear
\makeatother

\usepackage[framemethod=TikZ]{mdframed}
\mdfdefinestyle{MyFrame}{%
    linecolor=blue,
    outerlinewidth=1pt,
    roundcorner=10pt,
    innertopmargin=\baselineskip,
    innerbottommargin=\baselineskip,
    innerrightmargin=20pt,
    innerleftmargin=20pt,
    backgroundcolor=white}

\newcommand{\g}{$T$\xspace}
\newcommand{\m}{$M$\xspace}
\newcommand{\mo}{mechanism operator\xspace}
\newcommand{\mpc}{multiparty contract\xspace}
\newcommand{\ic}{initial contract\xspace}

\newcommand{\sel}{second-leg\xspace}
\newcommand{\fl}{first-leg\xspace}
\newcommand{\bs}{balance-sheet\xspace}
\newcommand{\cc}{central clearing\xspace}
\newcommand{\mln}{multilateral netting\xspace}

\newcommand{\rtm}{repo trading mechanism\xspace}
\newcommand{\fa}{financial asset\xspace}

\newcommand{\scu}{security\xspace}
\newcommand{\scus}{securities\xspace}
\usepackage[ruled,vlined,linesnumbered]{algorithm2e}
\usepackage{amssymb}
\usepackage[english]{babel}
\usepackage{amsthm}
\usepackage{bbm}
\usepackage{xcolor}
\usepackage{longtable}
\usepackage{booktabs, caption, makecell}

\usepackage{threeparttable}
\DeclareMathAlphabet{\mathpzc}{OT1}{pzc}{f}{it}
\usepackage{amsmath}
\DeclareMathOperator*{\argmax}{argmax}
\usepackage{appendix}

\theoremstyle{plain}
\usepackage{graphicx}
\usepackage{ulem}
\usepackage{enumitem}
\usepackage{tikz}
\usetikzlibrary{arrows.meta}
\usepackage{tikz,soul,setspace}

\usepackage{multicol}

\graphicspath{{./}} %

\usepackage{parskip}%
\usepackage{adjustbox}
\usepackage{lipsum}
\usepackage{tabularx}

\newtheorem*{proposition****}{Proposition 2: Comparative Balance-Sheet Impacts}
\newtheorem*{proposition***}{Proposition 1: Equivalence of Payoffs}
\newtheorem*{proposition*}{Proposition 2: Equal Attainable Trade Volume}
\newtheorem*{proposition**}{Proposition 3: Maximal Information}

\newtheorem*{observation*}{Observation 2: Payoffs when an Agent Defaults}

\newtheorem*{theorem*}{Theorem 8.1}
\newtheorem*{theorem**}{Theorem 8.2}
\newtheorem*{theorem***}{Theorem: A Computational Bound}
\newtheorem*{corollary*}{Corollary 1}
\newtheorem*{corollary**}{Corollary 2}
\newtheorem*{corollary***}{Corollary}

\newtheorem*{lemma*}{Lemma: Equal Flow Across Cuts}
\newtheorem*{lemma**}{Lemma: Flow Decomposition}

\newtheorem*{definition*}{Definition 1: Chains}
\newtheorem*{definition**}{The Replacement Repo Contract Guarantee Structure}
\newtheorem*{definition***}{Definition 2: Cycles}

\usepackage{multirow}
\usepackage{url}

\usepackage{float}

\usepackage{import}

\usepackage{geometry}
\usepackage{marginnote}
\usepackage{changepage}
\usepackage{xspace}
\usepackage{titlesec}

\usepackage{marvosym,soul,graphicx,%
xcolor,endnotes,amssymb,amsmath,amsthm,setspace,soul}
\definecolor{rune}{HTML}{4A6672}

\makeatletter
\renewcommand{\dt@brokenrefmark}[1]{}
\makeatother

\title{RepoMech: A Method to Reduce the Balance-Sheet Impact of Repo Intermediation}
\author{Daniel Aronoff\thanks{Research Affiliate, MIT Department of Economics}, Robert Townsend\thanks{Elizabeth and James Killian Professor of Economics, MIT Department of Economics} and Madars Virza \thanks{Research Scientist, MIT Media Lab Digital Currency Initiative} }

\date{November 11, 2026}

\begin{document}
\maketitle

\begingroup\renewcommand\thefootnote{}\footnotetext{\par \vskip 1em \noindent We wish to acknowledge LeAnn Tai for writing the code that implements the repo trading mechanism, and for preparing the visual representations in the paper, which are generated from the code; and Sriram Rajan for the many conversations that helped us to clarify our ideas. All errors are our responsibility.}

\addtocounter{footnote}{0}\endgroup

\pagenumbering{roman}

A repo trade involves the sale of a security coupled with a contract to repurchase at a later time. Following the 2008 financial crisis, accounting standards were updated to require repo intermediaries, who are mostly banks, to increase recorded assets at the time of the first transaction. Concurrently, US bank regulators implemented a supplementary leverage ratio constraint that reduces the volume of assets a bank is allowed record. The interaction of the new accounting rules and bank regulations limits the volume of repo trades that banks can intermediate. To reduce the balance-sheet impact of repo, the SEC has mandated banks to centrally clear all Treasuries trades. This achieves multilateral netting but shifts counterparty risk onto the clearinghouse, which can distort monitoring incentives and raise trading cost through the imposition of fees. We present RepoMech, a method that avoids these pitfalls by multilaterally netting  repo trades without altering counterparty risk.

\pagebreak
\tableofcontents
\pagebreak

\newpage
\pagenumbering{arabic}
\onehalfspacing

\section{Introduction}

Repurchase agreements (repos) involving U.S. Treasury securities are a cornerstone of short-term funding markets and Treasury market liquidity. The U.S. Treasuries repo market is intermediated by bank-affiliate broker-dealers. A repo trade involves a commitment to the sale of a security (the ``first-leg'') and the repurchase of the security at a later time (the ``second-leg''). In the United States, the Financial Accounting Standards Board (``FASB'') determines, inter alia, the accounting treatment of the first-leg transaction. Following the 2008 global financial crisis, FASB accounting standards were updated in a manner that, in practice, requires repo intermediaries to record a larger increase in assets at the first-leg. At the same time, U.S. bank regulators implemented a supplementary leverage ratio (``SLR'') that effectively reduced the volume of assets a bank may record for a given amount of capital.

The interaction of these accounting and regulatory changes---more repo assets recorded on the balance sheet at the first-leg and a leverage cap on total assets---has materially constrained the capacity of bank-affiliated broker-dealers to intermediate repo markets. A repo trade now causes a larger increase in recorded assets at the first-leg, which pushes the bank closer to its regulatory limit, than was the case before the reforms. This is so even when the bank acts as a matched-trade intermediary with offsetting repo loans and borrowings. Academic analyses have linked these constraints to reduced market-making and liquidity in repo and other fixed-income markets. For example, Duffie (2017) argues that post-crisis leverage requirements have raised the cost of repo intermediation and contributed to occasional stresses in U.S.\ money markets.

Responding to these concerns, in December 2023 the U.S. Securities and Exchange Commission (``SEC'') adopted a rule requiring banks to centrally clear, inter alia, their Treasuries repo trades. Central clearing reduces balance-sheet impact through multilateral netting. When trades are novated to a central counterparty (the ``CCP''), a dealer's repo borrowing and lending with multiple counterparties can be netted against each other on the CCP's books, leaving only a single net exposure per dealer. This netting can dramatically shrink gross balance-sheet exposures. To achieve this, however, central clearing reallocates counterparty credit risk to the CCP. By interposing a CCP, the new framework concentrates what used to be diverse bilateral risks onto a single entity. This concentration raises two related concerns. One is that the CCP becomes a central nexus of risk. The other is that agents no longer have incentives to monitor and evaluate the credit risk of their initial contract counterparties.

We present RepoMech, a repo trading mechanism that reduces the balance-sheet impact of repo intermediation by at least as much as central clearing, without introducing a new party and without altering the location of counterparty credit risk. RepoMech adapts the multilateral netting method in Aronoff et al.\ (2025) to repo trades by transforming the network formed by initial second-leg repo contracts into a set of chains and cycles on which traded objects flow along edges between nodes representing initial contracting agents. The initial contracts are terminated and replaced by multiparty contracts that net trades on each chain and cycle. For matched-trade intermediaries, this netting reduces the balance-sheet impact of repo intermediation from the value of the first-leg sale proceeds to the net profit earned by the intermediary at each leg (the ``total intermediation margin''), which ordinarily will be an order of magnitude smaller.

A crucial property of RepoMech is that the underlying object flows between initial repo counterparties are preserved. For each assignment of an initial second-leg trade to a chain or cycle, the same counterparties remain neighbors and the aggregate flow of objects between them is unchanged. When a node fails to perform, the affected chain or cycle is further decomposed into a set of chains, with the defaulting node placed into a bilateral contract with the neighboring node with whom it originally contracted. The new contract includes the flow of objects on the edge that connected them on the affected chain or cycle. No new counterparty exposures are created, and no initial counterparty exposures are lost. Thus, RepoMech combines three features: it reduces the balance-sheet impact of repo intermediation through multilateral netting, preserves contractual profit, and preserves the location of counterparty exposure.

\subsection{Related Literature}

Our work relates to a body of work that examines how accounting standards treat repos  and the implications for balance-sheet reporting. in particular the implications of accounting for the \fl as a final sale and the \sel as a forward contract, versus treating the entire trade as a secured financing. The latter implies a larger \bs impact than the former. Chang et.al. \cite{Chang-Repo105} provide a detailed post-mortem of Lehman’s Repo 105 transactions. Chircop et.al. \cite{Chircop2012} debate whether repos should be accounted as \fl final sales or secured financings, ultimately siding with the view that treating repos as financings provides a more faithful representation. Post-2008 FASB reforms (e.g. Accounting Standards Update 2011-03) largely resolved this debate by requiring repo trades to be accounted as secured financings with limited exceptions. Subsequent commentary by practitioners (\citep{Christodoulou2010,Pounder2011}) and accountants \citep{Hartwell-repo-105} describe how secured financing accounting  brought most repo assets onto balance sheets, preventing the kind of temporary balance-sheet “shrinkage” that Lehman and others had engineered. At the same time, researchers have evaluated the impact of bank capital and leverage regulations on repo activity. The introduction of leverage ratio requirements under Basel III has been widely cited as a key post-crisis development \citep{BIS2010}. Duffie \cite{Duffie2017} argues that the Supplementary Leverage Ratio (``SLR''), by penalizing low-risk, high-volume activities like repo intermediation, has reduced market liquidity and made banks less willing to intermediate these trades. Related empirical work shows that banks that intermediate the Treasuries repo market are operating close to the SLR lower bound \citep{Cochran2023} and that banks facing tighter leverage constraints cut back balance sheet-intensive positions around regulatory reporting dates, contributing to repo rate spikes \citep{munyan2015}. The link between leverage regulations and repo market capacity was vividly illustrated in March 2020, when surging demand for repo liquidity met strained dealer balance sheets; official reports and policymakers (Group of Thirty \cite{G30-2} and Fed \cite{FederalReserve2020}) noted that leverage constraints impeded dealers’ ability to absorb Treasuries, prompting temporary relief measures.

Finally, our paper speaks to the emerging literature on central clearing and its effects on repo markets. As regulators have promoted central clearing for repos, researchers have begun to analyze the benefits and risks of this shift. Central clearing achieves multilateral netting, which can relieve balance-sheet pressure on dealers – a point quantified by several studies. For instance, Hempel et al. \cite{Hempel2023} document that dealers already bilaterally net a large share of their bilateral repo positions where possible (by structuring offsetting trades), and they prefer the bilateral repo segment in part to take advantage of flexible margin and haircut terms. Kahn and Olson (2021) examine the participation in cleared repo and find that only the largest dealers and cash investors directly use clearing, while many smaller firms remain outside, suggesting barriers to access. Bowman et.al. \cite{bowman2024balancesheet} provide new evidence on the limited impact of central clearing on bank leverage ratios. They show that a significant portion of bilateral Treasury repo is already internally netted and thus would not further reduce balance-sheet usage even if moved to a CCP. On the other hand, Copeland and Kahn \cite{CopelandKahn2024SponsoredRepo} find that dealers do turn to centrally cleared repo (via sponsored clearing services) when their balance sheet space becomes scarce – for example, at quarter-end or when Treasury issuance surges. This behavior underscores that clearing’s netting benefits have tangible value for dealers under stress, but it also implies that dealers weigh those benefits against the costs of clearing (such as margin requirements and fees). A number of policy papers debate the net systemic effects of broad clearing mandates. The Treasury Market Practices Group \cite{TMPG-repo} and the Group of Thirty (2021) \cite{Group-of-30} have both recommended expanding central clearing in Treasury markets to bolster resilience. At the same time, industry participants have raised concerns: for example, Wuerffel \cite{Wuerffel2024} argues that \cc will raise trading costs which may shrink the volume of repo trades. 

Our work contributes to this debate by suggesting a third path that captures the balance-sheet netting advantages of clearing without altering counterparty risk or concentrating risk on a central party.

\subsection{Roadmap}

The remainder of the paper is organized as follows. Section \ref{sec:Repo Nomenclature} introduces key repo nomenclature and definitions that will be used throughout the analysis. Section \ref{sec:Repo Accounting Rules and Leverage Regulations} reviews U.S. repo accounting rules and bank leverage regulations in detail, contrasting the pre-2008 framework with post-crisis reforms and explaining how their interaction constrains repo intermediation. Section \ref{sec:RepoMech} presents the design of RepoMech, illustrating how it works through the chaining of repo transaction and the replacement of contracts to achieve multilateral netting. Section \ref{sec:Accounting and Legal Features of RepoMech} discusses the accounting and legal treatment of the proposed mechanism, explaining how it achieves a reduction in \bs impact. Section \ref{sec:Comparing RepoMech to Central Clearing} compares RepoMech to central clearing, showing how each affects balance-sheet exposures under current regulations and arguing that our mechanism can attain similar netting benefits without reallocation of counterparty risk. Section \ref{sec:conclusion-2} concludes with policy implications and suggestions for further research. Finally, Appendix \ref{app:Repo Accounting Rules} provides a detailed discussion of repo accounting rules.

\section{Repo Nomenclature}
\label{sec:Repo Nomenclature}

The following is a glossary of repo terms and notation we use in this paper.

\textbf{A Repo trade} A repo trade between two agents is comprised of two contracts entered into at the same time, a first-leg contract that closes immediately and a second-leg contract that closes at a later date. The first-leg of repo is a transaction in which agent $i$ purchases some number of units of the \scu, denoted $T$, from agent $j$ for unit price $p^{1}_{i\rightarrow j}$ with money, denoted $M$. The second-leg of repo is the forward transaction which occurs at a later date where agent $j$ (re)purchases $T$ from agent $i$ for unit price $p^{2}_{j\rightarrow i}$ with $M$. The repo rate is $r_{ij} = (p^{2}_{j\rightarrow i} - p^{1}_{i\rightarrow j})/p^{1}_{i\rightarrow j}$. This is the rate of return earned by $i$ and paid by $j$. $\{T_{ij},p^{1}_{i\to j}, p^{2}_{j\to i}\}$ denotes the elements of the repo trade between agents $i$ and $j$. By convention, the agent that receives money at the first-leg is the "repo borrower" and the other agent is the "repo lender". In this example agent $j$ is the repo borrower and agent $i$ is the repo lender.

\textbf{Treasuries classes} When the \scu is Treasures, trades are often partitioned into classes of Treasuries denoted by CUSIP, which are Treasuries that share the same coupon, issuance and maturity dates. Consequently, $T$ is a member of a class of interchangeable objects, meaning that the repo borrower is required to provide a $T$ from the designated set in the second-leg transaction. The alternative case is called "general collateral repo" where a borrower can repay with Treasuries from a set of designated CUSIP's. 

\textbf{Repo haircut} The repo haircut is the discount below the secondary market price of collateral, denoted $p_{T}$, paid by the first-leg buyer; $p^{1}_{i\to j} < p_{T}$.

\textbf{Some terminology}. \textit{\scu} is interchangeable with "Treasuries" (a specific type of \scu) and $T$. \textit{Financial objects} are money and the \scu. \textit{Nonperformance} is a failure to send financial objects required by a contractual obligation. \textit{Node} is a representation of an agent on a graph. A node that is formed by splitting an agent's inflows and outflows of second-leg $T$ is the child of the agent. A node formed by splitting a child node is a grandchild of the agent. We refer to "agent" and "node" interchangeably. \textit{DVP} is delivery versus payment, which applies to repo trades that requires the delivery of the \scu in exchange for payment of money, where the trade protocol specifies that each counterparty only receives its financial object when the other receives its financial object.

\textbf{Rehypothecation} The $T$ collateral is rehypothecatable, which means that a repo lender can sell the $T$ it receives from a counterparty at the first-leg. Rehypothecation enables the movement of collateral along a repo chain. At the \fl the \scu sent by a repo borrower is passed through one or more intermediaries until it reaches the repo lender. The repo lender sends money to the repo borrower in the opposite direction (with intermediaries possibly adding or subtracting amounts). At the \sel the flows are reversed. The intermediaries are simultaneously repo borrowers and lenders. We denote the ultimate repo lender a $UL$ (think money market fund), the ultimate repo borrower $UB$ (think risk manager e.g. a hedge fund) and the intermediary $INT$ (think balanced trade). Figure \ref{fig:First and Second-Leg Repo Chain} displays a rehypothecation chain with movement of financial objects $M$ and $T$ at each leg. Note that the volume of $T$ is fixed, but the volume of $M$ is variable. For example, the $M$ that moves from $INT_{h}$ to $INT_{i}$ at the \fl is $Tp^{1}_{h\to i}$, and the $M$ that moves in the opposite direction at the \sel is $Tp^{2}_{i\to h}$.

\textbf{Intermediation} We define repo intermediation as the activity of rehypothecating a fixed volume of the \scu (``matched-trades''). In Figure \ref{fig:First and Second-Leg Repo Chain} agents $BD_{h}, INT_{i}$ and $INT_{j}$ are engaged in intermediation. We use the terms ``intermediation'' when referring to the activity and ``intermediary'' when we refer to the entity (which might also be engaged in excess repo lending or borrowing, which is not intermediation).

\begin{figure}[H]
\centering
{\footnotesize{
 \textbf{First-leg}: $UL \underset{\text{$INT_{h}$ sale of $T$ to $UL$}}{\underbrace{
   \begin{aligned}
        T & \longleftarrow\\
        \longrightarrow & M
    \end{aligned}
}} INT_{h}\underset{\text{inter-dealer sale of $T$}}{\underbrace{
    \begin{aligned}
        T & \longleftarrow\\
        \longrightarrow & M
    \end{aligned}
}} INT_{i}\underset{\text{inter-dealer sale of $T$}}{\underbrace{
    \begin{aligned}
        T & \longleftarrow\\
        \longrightarrow & M
    \end{aligned}
}} INT_{j}\underset{\text{$UL$ sale of $T$ to $INT_{j}$}}{\underbrace{
    \begin{aligned}
        T & \longleftarrow\\
        \longrightarrow & M
    \end{aligned}
    }} UB$
}}
\end{figure}

\begin{figure}[H]
\centering
{\footnotesize{
\textbf{Second-leg}: $UL\underset{\text{ $UL$ sale of $T$ to $INT_{h}$}}{\underbrace{
   \begin{aligned}
        \longrightarrow & T\\
        M &\longleftarrow
    \end{aligned}
}} INT_{h}\underset{\text{inter-dealer sale of $T$}}{\underbrace{
    \begin{aligned}
        \longrightarrow & T\\
        M & \longleftarrow
    \end{aligned}
}} INT_{i}\underset{\text{inter-dealer sale of $T$}}{\underbrace{
    \begin{aligned}
        \longrightarrow & T\\
        M & \longleftarrow
    \end{aligned}
}}INT_{j}\underset{\text{ $INT_{j}$ sale of $T$ to $UB$}}{\underbrace{
    \begin{aligned}
        \longrightarrow & T\\
        M & \longleftarrow
    \end{aligned}
    }}UB$
\caption{First and Second-Leg Repo Chain}
\label{fig:First and Second-Leg Repo Chain}
}}
\end{figure}

\section{Interaction of Repo Accounting Rules and Bank Leverage Regulations}
\label{sec:Repo Accounting Rules and Leverage Regulations}

In this section we explain how the interaction of reforms to repo accounting rules and bank leverage regulations enacted after the 2008 financial crisis limit the volume of repo trades that bank affiliated broker-dealers can intermediate. Repo accounting rules were reformed to close loopholes that enabled an initial owner of a \scu to conceal its ownership. The loophole was closed by requiring the initial owner of \scus involved in a repo trade to retain the asset on its balance-sheet after the \fl sale. However, the change in the accounting treatment of repo trades requires intermediaries on the repo chain through which the \scus flow to increase their recorded assets.

Around the same time, bank leverage regulations were reformed to increase the minimum amount of capital, as a percentage of assets, that a bank is required to hold. The objective was to ensure a bank has adequate capacity to absorb losses on loans and other assets. One new leverage regulation in particular affects the repo market. It is the supplementary leverage ratio ("SLR") regulation, which sets a lower bound on the ratio of capital to unweighted  balance-sheet assets and off-balance-sheet derivatives exposures.\footnote{Walter \cite{Walter2019} reviews, inter-alia, the changes to bank capital regulations since the 2008 financial crisis.} In combination, the changes in accounting rules and leverage regulations have reduced the maximum volume of  repo that bank affiliated  broker-dealers can intermediate.

Section \ref{subsec:Repo accounting rules} provides a comparison of the balance-sheet impact of pre- and post-reform accounting rules on repo intermediation without references to the underlying rules themselves. Appendix \ref{app:Repo Accounting Rules} contains a more detailed discussion of repo accounting, with references to the applicable rules. Section \ref{subsec:Bank Leverage regulations on repo} explains the SLR and its importance to repo. Section \ref{Unintended Consequences} displays the effect of the interaction between the accounting and banking reforms. 

\subsection{Repo accounting rules}
\label{subsec:Repo accounting rules}

Prior to the post-crisis accounting changes, some repo transfers could qualify for sale accounting, with the repurchase commitment accounted for separately. Post-crisis FASB updates greatly narrowed that treatment by removing the former collateral-maintenance exception in ASC 860-10-40-24(b). Standard repos that maintain effective control are accounted for as secured borrowings, and repurchase-to-maturity transactions were specifically required to be treated as secured borrowings beginning with ASU 2014-11.

\subsubsection{Pre-reform accounting rules}

Prior to the reforms, there were two ways an agent could conceal its ownership or risk exposure to \scus by selling it at the \fl of a repo trade. One way applied when the \sel was scheduled before the maturity date of the \scus. In that case it was possible to treat the \fl of a repo trade as a final sale and the \sel as a forward sale (at the time of the \fl). Under final sale repo accounting the \scus of the seller are removed from the balance-sheet and the money received is added. The opposite holds for the purchaser. The \fl margin is the difference between them. The \sel transactions is recorded at fair-market-value ("FMV").\footnote{FMV is the \sel margin modified by counterparty credit risk and value-at-risk.} Notably, the \scu is removed from the balance-sheet of the \fl seller and recorded on the balance-sheet of the \fl buyer. Since the \scu is recorded onto, and then removed from, the balance-sheet of each agent as it moves across the repo chain at the \fl, a repo intermediary does not retain or record the asset in its balance-sheet. The balance-sheet impact on an intermediary is the sum of its first and \sel borrowing and lending margins (the "total intermediation margin"). Figure \ref{fig:First-Leg Balance-Sheet Impact of Repo Intermediation: Pre-Reform}  displays the pre-reform accounting of repo intermediation of volume $T$ of the \scu for agent $BT_{i}$,where $INT_i$ sells a quantity $T$ of the security to its left neighbor $INT_h$ and purchases the same quantity from its right neighbor $INT_j$ at the first leg..  

\begin{figure}[H]

\begin{center}
\text{$INT_i$ intermediation of $T$ on the repo chain in Figure \ref{fig:First and Second-Leg Repo Chain}}\vskip5pt
\begin{tabular}{l|l}
\hline
\multicolumn{1}{c|}{\textbf{Liabilities}} & \multicolumn{1}{c}{\textbf{Assets}} \\
\hline
& \underline{First-leg net cash margin}\\
& $+\,Tp^{1}_{h\to i} - Tp^{1}_{i\to j}$\\
& \underline{Second-leg commitments}\\
& $\pm\,\mathrm{FMV}(\text{purchase from }h)$\\
& $\pm\,\mathrm{FMV}(\text{sale to }j)$\\
\hline
& $\Delta A \approx$ net intermediation margin $\pm$ net FMV\\
\end{tabular}
\end{center}
\caption{Balance-sheet impact at the first-leg date under pre-reform final-sale accounting}
\label{fig:First-Leg Balance-Sheet Impact of Repo Intermediation: Pre-Reform}
\end{figure}

Before ASU 2014-11, repurchase-to-maturity transactions were generally analyzed under the sale-accounting framework in Topic 860 and often were accounted for as sales with forward repurchase commitments when the derecognition conditions were met. In practice, that treatment could permit economically important exposure to remain outside the balance-sheet presentation that would apply under secured-borrowing accounting.

Two notable instances of perceived abuse led to changes in accounting rules to prevent this practice. In one instance, Lehman Brothers devised a transaction structure, called Repo 105, which it employed around financial disclosure dates for several years prior to its 2008 bankruptcy. The maneuver enabled Lehman to conceal billions of dollars of subprime mortgage exposure by recording the sale of subprime securities at the \fl as a final sale. The trade was timed so that the financial reporting date fell in-between the first and \sel{s}, which enabled Lehman to report a balance-sheet that did not contain the traded subprime securities  \citep{Pounder2011}.\footnote{It is unclear how prevalent was the use of final-sale accounting prior to the reforms. In March 2010 the SEC sent out a “Dear CFO” letter to 19 banks and financial institutions in the US asking how they accounted for repos. The responses indicated that the surveyed institutions did not use final-sale accounting for a majority of their repo trades. However, the survey did not cover the entire universe of repo participants and the respondents did not provide an exact percentage breakdown of repo accounting treatment \citep{Christodoulou2010}.} In another instance, MF Global concealed its exposure to billions of dollars of low-rated sovereign debt by structuring repo-to-maturity trades.  MF Global was required to cover any shortfall caused by a default on the sovereign debt, but the indemnification was off balance-sheet and unreported. Appendix \ref{app: Lehman and MF Global Repo Strategies} discusses the strategies employed by Lehman and MF Global.

\subsubsection{Post-reform accounting rules}

 The accounting reforms that were enacted subsequent to the perceived Lehman and MF Global abuses require, in practice, that nearly all repo trades are treated as secured financings whereby the \scu remains on the balance-sheet of the initial owner even after it is sold at the \fl. This changes the \fl balance-sheet impact of repo intermediation. Under secured-borrowing accounting, the transferor recognizes the cash received and a corresponding borrowing liability, while the transferred security remains on the initial owner's balance sheet. The liability then accretes to the repurchase amount over the term. The transferee derecognizes the cash paid, records a receivable, and does not recognize the transferred security absent default by the transferor. 

 The repo lender replaces its \fl cash outflow used to purchase the \scu with a receivable of equal value. Otherwise, its \bs does not change.\footnote{The gap in price between first and \sel is treated as an interest accrual which does not appear at the first-leg.}. The \bs impact on an intermediary is the combination of the two, which reduces in value terms to the impact of its borrowing. Compared to pre-reform, this shift from total intermediation margin to first-leg sale price increase in asset value, constitutes an order of magnitude increase in assets associated with repo intermediation. Figure \ref{fig:First-Leg Balance-Sheet Impact of Repo Intermediation: Post-Reform} displays the post-reform accounting of repo intermediation of volume $T$ of the \scu for agent $INT_{i}$ on the repo chain in Figure \ref{fig:First and Second-Leg Repo Chain}.

\begin{figure}[H]

\begin{center}
\text{$INT_i$ intermediation of $T$ on the repo chain in Figure \ref{fig:First and Second-Leg Repo Chain}}\vskip5pt
\begin{tabular}{l|l}
\hline
\multicolumn{1}{c|}{\textbf{Liabilities}} & \multicolumn{1}{c}{\textbf{Assets}} \\
\hline
+ secured borrowing obligation & + cash received on borrowing leg \\
(initially = cash received; & $(= Tp^{1}_{h\to i})$ \\
accretes to repurchase amount) & \\
\hline
& $\Delta A \approx +\,Tp^{1}_{h\to i}$\\
\end{tabular}
\end{center}
\caption{Net first-leg balance-sheet impact of repo intermediation under secured-borrowing accounting}

\label{fig:First-Leg Balance-Sheet Impact of Repo Intermediation: Post-Reform}
\end{figure}

\subsection{Bank leverage regulations}
\label{subsec:Bank Leverage regulations on repo}

Bank holding companies and their deposit-taking subsidiaries are independently subject to capital and leverage regulations. The rules do not directly apply to broker-dealer subsidiaries of banks who trade in the repo market. Consequently, we focus on the impact of repo trading on the consolidated balance-sheet of the bank holding company affiliate of the broker-dealer which we refer to as a "bank".  

\subsubsection{Pre-reform capital rules}

For many decades banks operating in the U.S. have been subject to a number of minimum capital ratio requirements (CRS \cite{CRS2023} Table 3). The numerator of each ratio is a different measure of bank capital and denominator of each is the same risk weighted assets ("RWA"). A risk weight, which is encoded in regulation, is the percentage of an asset's value that is represented in the denominator. A U.S. Treasury is assigned a risk weight of $0\%$. The RWA is the asset value multiplied by the risk weight. This implies that U.S. Treasuries do not appear in the denominator of any risk-weighted capital ratio. In addition, banks have been subject to a minimum leverage ratio requirement, where the denominator is composed of unweighted assets (i.e. assets enter at their values recorded on the balance-sheet), which includes U.S. Treasuries. 

\subsubsection{Post-reform capital rules}

The 2008 financial crisis elicited regulatory reforms designed to remedy perceived flaws in banking regulations that contributed to the crisis. A key reform was Basel III, which was enacted by the Bank for International Settlements \citep{BIS2010} and adopted by U.S. banking regulators. A key aim of Basel III was to reduce the risk of bank insolvency by adding off balance-sheet exposures, such as derivatives, to the denominator of the leverage ratio ( \cite{Walter2019}. The SLR places a lower bound, denoted by $\underline{L}$, on the ratio of bank capital to unweighted assets plus off balance-sheet exposures, of 3\% with an additional 2\% for large globally systemically important banks (''GSI'''s) \cite{BIS2010}.

\begin{center}
SLR: capital/(assets + exposures) $\geq$ 3\% + 2\% for GSIB's = $\underline{L}$
\end{center}

An important observation is that the SLR is more restrictive than the pre-existing minimum leverage ratio. The lower bound of the leverage ratio in each case is comparable, at 3-5\%. However, the inclusion of exposures in the denominator of the SLR implies that more capital is required to achieve a given ratio under the SLR.  In recent years a number of banks with the largest share of repo intermediation have been operating near their SLR lower bounds. This has prompted concern that the SLR regulation has placed a binding constraint which is limiting the capacity of their broker-dealer affiliates to intermediate the US Treasuries cash and repo markets \citep{Duffie2017}Figure 2.1.1 ). Economist Darrell Duffie expressed concern over the restrictive impact of the SLR on intermediation in his 2018 Baffie Lecture.

\begin{quote}
The concern is instead that the amount of intermediation provided by banks to low-risk asset markets has become inefficiently low...one can infer from Figure 2.1.1 that the largest U.S. dealer banks must carefully consider the impact of the leverage ratio rule (SLR) on their minimum capital levels when deciding how much of their balance sheet to allocate to safe asset intermediation (Duffie \cite{Duffie2017}  Chapter 2).
\end{quote}

The Board of Governors of the Federal Reserve System expressed concern over the restrictive impact of the SLR on intermediation in the Congressional Record in 2020.

\begin{quote}
Large holding companies have cited balance sheet constraints for their broker-dealer subsidiaries as an obstacle to supporting the Treasury market. Specifically, the \textbf{supplementary leverage ratio} can limit holding companies' ability to own Treasuries outright \citep{FederalReserve2020}.
\end{quote}

The Group of 30, which is comprised of leading academic, regulator and financial industry leaders cite the SLR's impact on the repo market as a fundamental source of financial dysfunction.

\begin{quote}
With leverage ratios, especially the SLR, currently the binding regulatory constraint on capital allocation at many of these banks, they are discouraged from allocating capital to market intermediation in the Treasury markets and especially in the Treasury repo markets, the liquidity of which is critical to all dealers in Treasury securities and other leveraged providers of Treasury market liquidity.\citep{G30-2}
\end{quote}

\subsection{Intermediation constraints}
\label{subsec: Intermediation constraints}

The interaction of the post-reform increase in balance-sheet impact of repo intermediation with the tightened leverage rule has two interrelated effects. One effect is that repo intermediation pushes downward the leverage ratio and the SLR toward their lower bounds at a faster rate compared to pre-reform. This follows from the order of magnitude larger balance-sheet impact of repo intermediation post-reform (Figure \ref{fig:First-Leg Balance-Sheet Impact of Repo Intermediation: Post-Reform}) compared to pre-reform (Figure \ref{fig:First-Leg Balance-Sheet Impact of Repo Intermediation: Pre-Reform}). The other effect is that banks active in the repo market are operating closer to their SLR lower bound and must increase capital devoted to repo in order to increase intermediation volume (Duffie \cite{Duffie2017}  Chapter 2). The imposition of the SLR on top of the pre-existing capital and leverage ratio lower bounds can be viewed as a marginal increase in the lower bound of a composite capital and leverage ratio. Looked at this way, it is natural to ask what effect a marginal increase in the lower bound has on the volume of repo intermediation a bank will undertake. We evaluate each effect.

\textbf{Effect of repo volume on the SLR} Figure \ref{fig:First-Leg Balance-Sheet Impact of Repo Intermediation: Post-Reform} shows that repo intermediation of volume $T$ of the \scu increases balance-sheet assets, which we denote $\Delta A$. Intermediation also increases liabilities. Equity - which is part of capital- is adjusted by the gap between the increase in liabilities and assets. The gap represents the total intermediation margin. There is no apriori way to determine the sign of the margin, however consideration of the thin margins observed in the repo market suggest the size of the gap is small relative to $\Delta A$.\footnote{The SOFR index of overnight U.S. Treasuries repo rates has been close to the Fed funds target rate except during times of market disruption. The margin for an overnight loan priced at an annualized interest rate of e.g. 5\% is very small relative to the volume of the trade. See Federal Reserve Bank of New York  \cite{FRBNY2023}} Therefore, we will ignore it. Equation \ref{eq:incremental leverage} shows that an increase in repo borrowing pushes down the leverage ratio by increasing assets in the denominator. The ratio cannot drop below the SLR bound $\underline{L}$.  If Equation \ref{eq:incremental leverage} starts out as an equality, an increase in $\underline{L}$ will reduce repo intermediation.

\begin{equation}
\label{eq:incremental leverage}
\underline{L} \leq \underset{\text{incremental repo $\downarrow$ SLR}}{\underbrace{\frac{capital}{assets + exposures + \Delta A}}}
\end{equation}

\textbf{Adjusting capital at the SLR lower bound} A bank at the SLR lower bound can enable its  broker-dealer affiliate to increase repo volume by allocating more capital to the broker-dealer. This can be accomplished by raising additional capital, which will increase the numerator (and RHS) of Equation \ref{eq:incremental leverage}, or by re-allocating internal capital to the broker-dealer, which leaves the RHS of Equation \ref{eq:incremental leverage} unchanged. It is not possible to predict whether a bank will find it profitable to do either of these things. However, it is possible to show that, for any repo margin and cost of capital and concave repo lender demand function, increasing the SLR lower bound $\underline{L}$ decreases the repo volume traded.\footnote{In Appendix \ref{app:Concavity} we derive an empirically founded increasing concave repo demand function, $D(r_{UL})$, for a money market fund repo lender and a decreasing concave supply function $S(r_{rm})$ for a hedge-fund borrower. A similar argument as made in the text can be used to show that an increase in the SLR lower bound will reduce repo lending to hedge fund borrowers.}

We use the notation from Section \ref{sec:Repo Nomenclature} Suppose $INT_{j}$ is a subsidiary of a bank holding company. Its affiliate bank can increase repo volume by allocating more capital to $INT_{h}$. This can be accomplished by raising additional capital, which will increase the RHS of Equation \ref{eq:incremental leverage}, or by re-allocating internal capital to $INT_{h}$, which leaves the RHS of Equation \ref{eq:incremental leverage} unchanged. Either method involves an opportunity cost of capital for repo, which we denote by the unit cost $c$.\footnote{The opportunity cost of external capital is the market price of acquiring the capital. The opportunity cost of internal capital is the profit that could be earned by deploying the capital elsewhere in the bank.} We analyze the effect that imposing the SLR lower bound on attainable trade volume between $INT_{h}$ and $UL$ in Figure \ref{fig:First and Second-Leg Repo Chain} when $INT_{h}$ is operating at the SLR lower bound. To achieve a unit increase in trade volume requires $INT_{h}$ increase capital by $\underline{L}$ units.\footnote{The capital requirement is determined by $\frac{\Delta capital\; +\; capital}{1 + assets+ exposures} = \underline{L}$.} The marginal cost of capital to enable an increase in repo volume at the SLR lower bound is $c\underline{L}$. $D(r_{UL})$ denotes the volume of $T$ that  $UL$ is willing to trade at repo rate $r_{UL}$. $r_{int}$ is the market-determined inter-dealer repo rate at which $INT_{h}$ can enter into a repo trade to acquire the collateral $T$ it is required to send to $UL$ at the \fl. We model $INT_{h}$ as choosing the repo rate. $INT_{h}$'s problem is to set the repo rate $r_{UL}$ it offers to $UL$ to maximize its profit. Equation \ref{eq:$INT_{h}$ problem with SLR} is $INT_{h}$'s decision problem.

\begin{equation}
\label{eq:$INT_{h}$ problem with SLR}
\argmax_{r_{UL}}\underset{\text{trading profit}}{\underbrace{(r_{int} - r_{UL})D(r_{UL})}} - \underset{\text{capital cost}}{\underbrace{ c\underline{L}\cdot \text{max}\{D(r_{UL}) - \overline{D},0\}}}
\end{equation}

Where $\overline{D}$ is the trade volume at the SLR constraint. The first-order equilibrium condition of Equation \ref{eq:$INT_{h}$ problem with SLR} at the SLR lower bound is

\begin{equation}
\label{eq: FOC with SLR}
f = \big[(r_{int} - r_{UL}) -  c\underline{L}\big]dD/dr_{UL} - D(r_{UL}) = 0
\end{equation}

The second derivative is 

\[-D/dr_{UL} +  \big[(r_{int} - r_{UL}) -  c\underline{L}\big]d^{2}D/dr_{UL}dr_{UL} - dD/dr_{UL} < 0\]

Which implies $INT_{h}$ faces a concave decision problem with a unique solution.\footnote{ The same result obtains if $UL$ chooses the repo rate. $UL$'s problem is $\argmax_{r_{UL}}D(r_{UL})\; s.t.\;(r_{int} - r_{UL})D(r_{UL}) - c\underline{L}\cdot\text{max}\{D(r_{UL} - \overline{D}, 0\}$.The KKT conditions ensure that the Lagrangian multiplier $\lambda \geq 0$. The remainder of the derivation is left as exercise for the reader.} It is not possible to infer whether, in a given instance, the maximum lies at a client repo rate that places $INT_{j}$ above or below the initial SLR lower bound. However, we can apply the implicit function theorem to evaluate the change in repo rate $r_{UL}$ that is induced by an increase in the SLR lower bound $\underline{L}$ when the constraint is binding.

\begin{equation}
\label{eq: SLR effect}
\frac{dr_{UL}}{d\underline{L}} = \frac{-df/d\underline{L}}{df/dr_{UL}} < 0
\end{equation}

Equation \ref{eq: SLR effect} demonstrates that an increase in the SLR lower bound $\underline{L}$ induces a decrease in $r_{UL}$, which reduces repo volume $D(r_{UL})$ for banks operating at the lower bound. By analogy, this shows that the SLR lower bound reduced the volume of repo intermediation for those banks that are operating near the SLR lower bound.

\subsection{Unintended effect of accounting rules and leverage regulations on  intermediation capacity}
\label{Unintended Consequences}

The reforms to repo accounting were designed to close loopholes that enable a bank to conceal ownership of or exposure to \scus to regulators and investors. The SLR was designed to provide a hard backstop to bank insolvency risk. The different aims of the two sets of regulations address distinct risks. However, in combination, they reduce the capacity of banks to intermediate repo trades. The accounting rules that require the \sel obligations to be recorded on the balance sheet increase recorded assets associated with a repo trade, which pushes down the leverage ratio. At the same time the SLR increases the lower bound underneath the leverage ratio.

\section{RepoMech }
\label{sec:RepoMech}

In this section we adapt the TradeMech method of netting trades in Aronoff et.al. \cite{Aronoff2025mech} and Aronoff \cite{Aronoffnetwrap} to repo trades. Table \ref{tab:Initial Contracts} displays initial \fl and \sel contracts between agents. Each row displays the elements of the trades between counterparties. The first column is the number assigned to the repo trade. The second column is the repo lender (sends \m in the \fl and sends \g in the \sel). The third column is the repo borrower (sends \g in the \fl and sends \m in the \sel). The fourth column is the \fl unit price of \g. The fifth column is the \sel unit price of \g.  The sixth column is the traded units of \g, which is the same for the \fl and \sel.

\begin{table}[H]
\centering
\begin{tabular}{l|ccccc}
\toprule
\makecell{\textbf{Trade}\\\textbf{Number}} &\makecell{\textbf{Firm}\\\textbf{ Id}}& \makecell{\textbf{Counter-}\\\textbf{ party Id}} & \makecell{\textbf{First-Leg}\\\textbf{ Price}} & \makecell{\textbf{Second-Leg}\\\textbf{ Price}} & \makecell{\textbf{Collateral}\\\textbf{ Units (T)} }\\
\midrule
1 & h & i & \$4.90  & \$5.25 & 5  \\
2 & k & i & \$5.80  & \$6.30  & 3  \\
3 & i & j & \$6.10  & \$6.55 & 5  \\
4 & i & g & \$3.00    & \$3.00    & 4  \\
5 & g & j & \$5.40  & \$5.95 & 10 \\
6 & l & g & \$5.40  & \$5.95 & 6  \\
7 & h & f & \$3.00    & \$3.30  & 10 \\
8 & f & h & \$3.00    & \$3.10  & 8  \\
9 & k & g & \$2.90  & \$3.77 & 8  \\
10 & g & f & \$6.22 & \$6.53 & 10 \\
11 & f & i & \$4.60  & \$5.12 & 6  \\
\bottomrule
\end{tabular}
\caption{Initial Repo Contracts}
\label{tab:Initial Contracts}
\end{table}

\subsection{Assignment of trades to chains}
\label{subsec:Assignment of trades to chains}

After the initial contracts are executed  and the \fl contracts are cleared and settled, \sel trades between  agents are netted on collateral \g. The \sel netting works as follows. In contract \#7 agent $h$ sends 10 units of \g to agent $f$. In contract \#8 agent $f$ sends 8 units of \g to agent $h$. In the net trade agent $h$ sends agent $f$ 2 units of \g. The total \m paid by $f$ in contract \#7 is \$30.30. The total \m paid by $h$ in contract \#8 is \$24.80. The net price paid by $f$ to $h$ is \$5.50. The net unit price is \$2.75. Figure \ref{fig:Initial S Flow Network} depicts the network of \g flow formed by the netted initial contracts. Nodes represent agents and the numbers on the directed edges represent the units of \g flowing between the agents.\footnote{ We sometimes represent the flow from e.g. $h$ to $f$ by  $2_{h\to f}$ or, abstractly as $\g_{h\to f}$.}

\begin{figure}[H]
\begin{center}
\includegraphics[page=1,width=346pt,height=250pt,keepaspectratio]{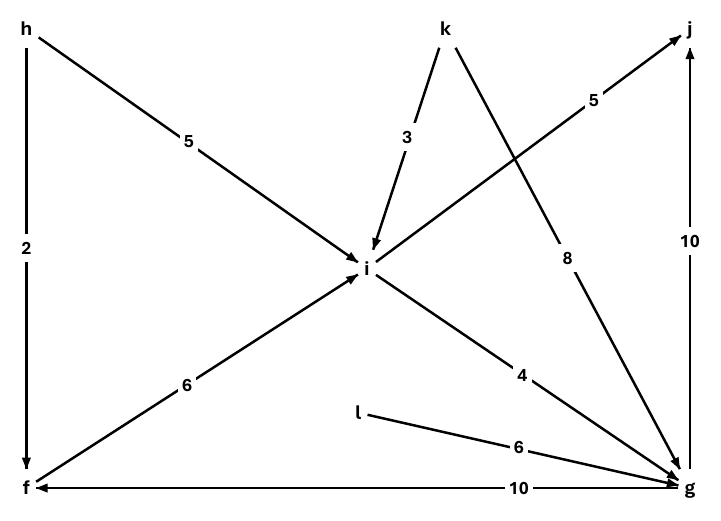}
\end{center}
\caption{Initial  \g - Flow Network}
\label{fig:Initial S Flow Network}
\end{figure}

\subsection{Separate excess trade flows and matched-trade trade flows}         
\label{subsec: Separate excess trade flows and matched-trade trade flows}

The second step divides each node into at most two nodes; one node has equal inflow and outflow of \g (the "intermediary node" or "matched-book trade node") and the other node has the excess of inflow or outflow (the "excess flow node"), if any. Figure \ref{fig:Splitting node f with a net inflow of 6} first splits parent node $f$, which has a net inflow of 6, into the intermediary node $INT_f$ in the middle and the excess-inflow node $UB_f$ on the right. Figure \ref{fig:Splitting node g with a net outflow of 2} then starts from that updated network and splits parent node $g$, which has a net outflow of 2, into the intermediary node $INT_g$ below $INT_f$ and the excess-outflow node $UL_g$ on the left. The split of $f$ is retained while $g$ is split. The original $f$ and $g$ nodes are the "parent" nodes, and the intermediary and excess flow nodes into which they are divided are their "child" nodes. The parent nodes and replaced edges are faded; the child nodes and current connections are black.

A key property of the node split is that the flows between child nodes of $f$ and $g$ are the same as flow between the parent $f$ and $g$ in their netted initial contracts. Two other features are first, the trade pattern for the children of $f$ and $g$ is unaffected by the order in which the splitting occurs. Second, the selection of \g flows that are attached to the excess outflow and inflow nodes $UL$ and $UB$, respectively, in steps 1 and 1' of the Nodes Splitting Algorithm  are optimized for the un-netted excess flows by assigning \g flows to the $UL$ and $UB$ nodes in ascending order of the associated  \fl cash inflow. This minimizes \bs impact for a repo borrower - who has a \sel excess outflow of \g -  for whom \bs assets increase by the amount of \fl cash received. For a repo lender, who does not incur a \bs impact, the algorithm minimizes its \fl cash outflow.\footnote{Section  \ref{sec:Repo Accounting Rules and Leverage Regulations} explains the \bs impact of repo.}

\begin{figure}[H]
\begin{center}
\includegraphics[page=1,width=346pt,height=250pt,keepaspectratio]{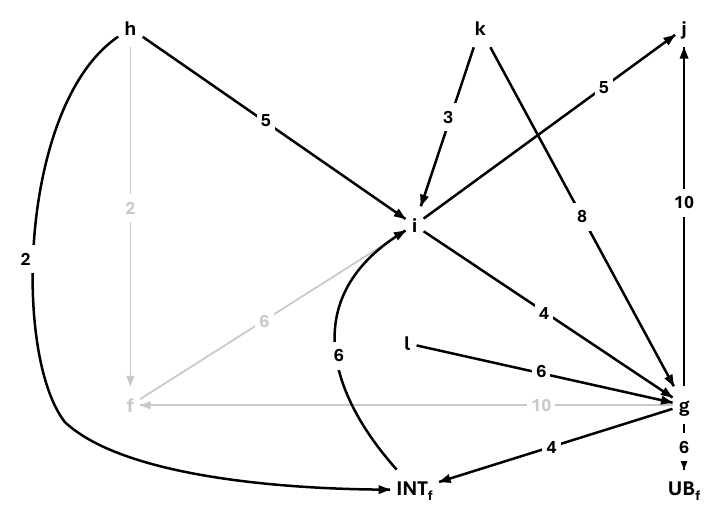}
\end{center}
\caption{Splitting node $f$ with a net inflow of 6}
\label{fig:Splitting node f with a net inflow of 6}
\end{figure}

\begin{figure}[H]
\begin{center}
\includegraphics[page=1,width=346pt,height=250pt,keepaspectratio]{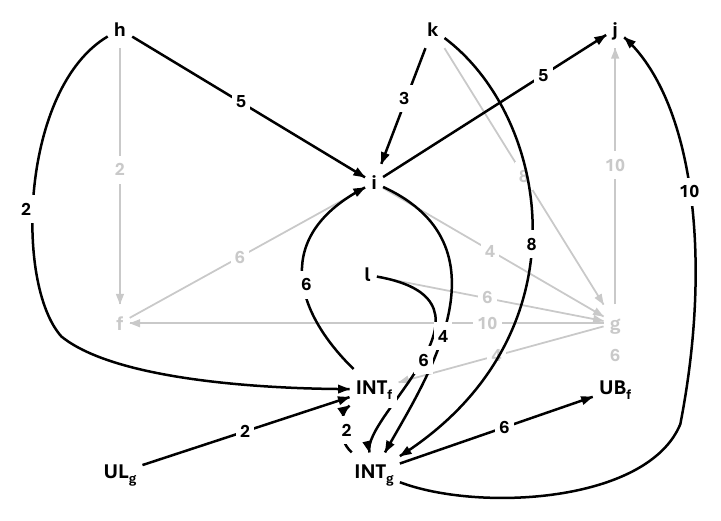}
\end{center}
\caption{Splitting node $g$ with a net outflow of 2}
\label{fig:Splitting node g with a net outflow of 2}
\end{figure}
\newpage

\underline{NODE SPLITTING ALGORITHM}
\textbf{For node $i$ with a net $T$ outflow $E$:}
\begin{enumerate}[label=\textbf{Step \arabic*.}, leftmargin=*, itemsep=0.4em]
\item  \textbf{Choose} $E$ volume of outflow in ascending order of \fl unit price.

\item Subtract the selected $E$ outflow from directed edges flowing from $i$ in the trade flow graph.

\item Create a new node $UL_{i}$ and attach directed edges carrying the $E$ outflow that was detached from $i$. Label the remaining node $INT_{i}$.
\end{enumerate}
  
  \textbf{For node $i$ with a net $T$ inflow $E$:}
  \begin{enumerate}[label=\textbf{Step \arabic*'.}, leftmargin=*, itemsep=0.4em] 
  \item \textbf{Choose} $E$ volume of inflow in ascending order of the \fl unit price\;
   \item Subtract the selected $E$ inflow from directed edges flowing into $i$ from other nodes\;
  \item Create a new node $UB_{i}$ and attach directed edges carrying the $E$ inflow that was detached from $i$. Label the remaining node $INT_{i}$
  \end{enumerate}

The node splitting results in the flow network of Figure \ref{fig:Trade Flow Network}. 

\begin{figure}[H]
\begin{center}
\includegraphics[page=1,width=346pt,height=250pt,keepaspectratio]{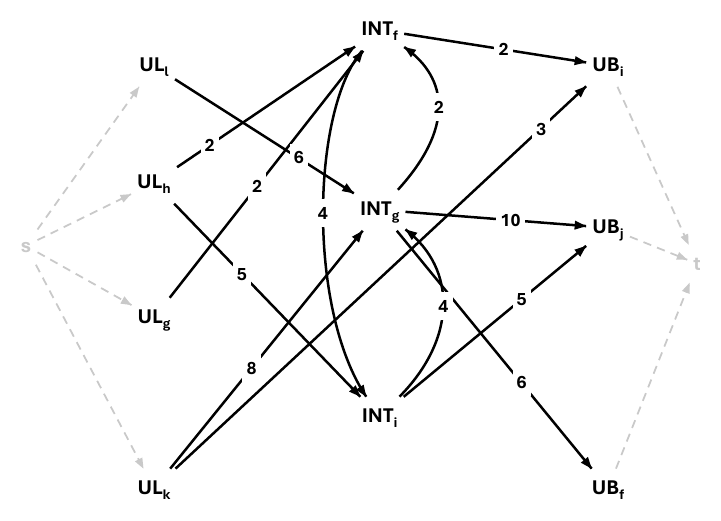}
\end{center}
\caption{Final trade flow network (TFN) after all node splits, continuing the $f$-then-$g$ sequence in Figures \ref{fig:Splitting node f with a net inflow of 6} and \ref{fig:Splitting node g with a net outflow of 2}.}
\label{fig:Trade Flow Network}
\end{figure}

The TFN is then decomposed into chains and cycles and the units of money are re-attached to the edges.\footnote{For details on the method of decomposing the TFN into chains and cycles Aronoff et.al. \cite{Aronoff2025mech}.}

{\small{
Chain 1: $UL_{k}$ 
\tikz[baseline]{
  \draw[->] (0,0.1) -- (1,0.1) node[midway, above] {3\g};
  \draw[<-] (0,-0.1) -- (1,-0.1) node[midway, below] {\$18.90\m};} $UB_{i}$

\vspace{1em}

Chain 2: $UL_{K}$ 
\tikz[baseline]{
  \draw[->] (0,0.1) -- (1,0.1) node[midway, above] {8\g};
  \draw[<-] (0,-0.1) -- (1,-0.1) node[midway, below] {\$30.16};} $INT_{g}$ 
\tikz[baseline]{
  \draw[->] (0,0.1) -- (1,0.1) node[midway, above] {8\g};
  \draw[<-] (0,-0.1) -- (1,-0.1) node[midway, below] {\$47.60};} $UB_{j}$

\vspace{1em}

Chain 3: $UL_{l}$ 
\tikz[baseline]{
  \draw[->] (0,0.1) -- (1,0.1) node[midway, above] {2\g};
  \draw[<-] (0,-0.1) -- (1,-0.1) node[midway, below] {\$11.90};} $INT_{g}$ 
\tikz[baseline]{
  \draw[->] (0,0.1) -- (1,0.1) node[midway, above] {2\g};
  \draw[<-] (0,-0.1) -- (1,-0.1) node[midway, below] {\$11.90};} $UB_{j}$

\vspace{1em}

Chain 4: $UL_{l}$ 
\tikz[baseline]{
  \draw[->] (0,0.1) -- (1,0.1) node[midway, above] {4\g};
  \draw[<-] (0,-0.1) -- (1,-0.1) node[midway, below] {\$23.8};} $INT_{g}$ 
\tikz[baseline]{
  \draw[->] (0,0.1) -- (1,0.1) node[midway, above] {4\g};
  \draw[<-] (0,-0.1) -- (1,-0.1) node[midway, below] {\$26.12};} $UB_{f}$

Chain 5: $UL_{h}$ 
\tikz[baseline]{
  \draw[->] (0,0.1) -- (1,0.1) node[midway, above] {5\g};
  \draw[<-] (0,-0.1) -- (1,-0.1) node[midway, below] {\$26.25};} $INT_{i}$ 
\tikz[baseline]{
  \draw[->] (0,0.1) -- (1,0.1) node[midway, above] {5\g};
  \draw[<-] (0,-0.1) -- (1,-0.1) node[midway, below] {\$32.75};} $UB_{j}$

\vspace{1em}

Chain 6: $UL_{h}$ 
\tikz[baseline]{
  \draw[->] (0,0.1) -- (1,0.1) node[midway, above] {2\g};
  \draw[<-] (0,-0.1) -- (1,-0.1) node[midway, below] {\$8.20};} $INT_{f}$ 
\tikz[baseline]{
  \draw[->] (0,0.1) -- (1,0.1) node[midway, above] {2\g};
  \draw[<-] (0,-0.1) -- (1,-0.1) node[midway, below] {\$10.24};} $UB_{i}$

\vspace{1em}

Chain 7: $UL_{g}$ 
\tikz[baseline]{
  \draw[->] (0,0.1) -- (1,0.1) node[midway, above] {2\g};
  \draw[<-] (0,-0.1) -- (1,-0.1) node[midway, below] {\$13.06};} $INT_{f}$ 
\tikz[baseline]{
  \draw[->] (0,0.1) -- (1,0.1) node[midway, above] {2\g};
  \draw[<-] (0,-0.1) -- (1,-0.1) node[midway, below] {\$10.24};} $INT_{i}$ 
\tikz[baseline]{
  \draw[->] (0,0.1) -- (1,0.1) node[midway, above] {2\g};
  \draw[<-] (0,-0.1) -- (1,-0.1) node[midway, below] {\$6.00};} $INT_{g}$ 
\tikz[baseline]{
  \draw[->] (0,0.1) -- (1,0.1) node[midway, above] {2\g};
  \draw[<-] (0,-0.1) -- (1,-0.1) node[midway, below] {\$13.06};} $UB_{f}$
}}

\vspace{2em}
\noindent
Cycle 1:
\begin{center}
\begin{tikzpicture}[>=stealth]

  \node (g) at (0, 0) {$INT_{g}$};
  \node (f) at (3.5, 0) {$INT_{f}$};
  \node (i) at (1.75, 3) {$INT_{i}$};

  \draw[<->] (g) -- (f) 
    node[midway, above, sloped] {2\g}
    node[midway, below, sloped] {\$13.06};

  \draw[<->] (f) -- (i) 
    node[midway, above, sloped] {2\g}
    node[midway, below, sloped] {\$10.24};

  \draw[<->] (i) -- (g) 
    node[midway, above, sloped] {2\g}
    node[midway, below, sloped] {\$6.00\m};

\end{tikzpicture}
\end{center}

The flow of objects between the nodes associated with agents is divided among chains and cycles, but the aggregate flows between pairs of agents are unaffected by transformation of the graph representing \sel trades.\footnote{Aronoff et.al. \cite{Aronoff2025mech} provides an overview of the mechanics of the process of iteratively replacing contracts on chains and cycles.} 

\subsection{Replacement contracts}
\label{subsec:Replacement contracts}

In this section we describe features of replacement contracts in Aronoff et.al. \cite{Aronoff2025mech} that are relevant to the accounting treatment of the transactions.\footnote{Aronoff \cite{Aronoffnetwrap} provides an algorithmic computational and cryptographic foundation for retaining the initial exposure map between counterparties, of which \citep{Aronoff2025mech} is an application.} 

\textbf{Replacement contracts} After clearing and settlement of \fl contracts, the initial \sel contracts are terminated and replaced by \mpc contracts on chains and cycles. For each node, the net flow is computed from its connected edges. It is the outflow minus the inflow. A node sends its negative net flow and receives its positive net flow, as directed by the \mo. Tables \ref{tab:Net flows on Chain 7} and \ref{tab:Net flows on Cycle 1} display he net flow of objects on a chain and a cycle.

\begin{table}[H]
\centering\
\small{
\begin{tabular}{l|c c c c c|c}
\textbf{Object flow} &\textbf{$UL_{g}$} & \textbf{$INT_{f}$}  & \textbf{$INT_{i}$} & \textbf{$INT_{g}$}  & \textbf{$UB_{f}$} & \makecell{Net flow\\ on Chain}\\
\hline
\g - flow   & 2 - out  & &  &   & 2 -in & 0\\
\m - flow   & 13.06 - in & 2.82 - out & 4.24 - out  & 7.06 - in  &  13.06 - out & 0\\              
\end{tabular}
}
\caption{Net flows on Chain 7}
\label{tab:Net flows on Chain 7}
\end{table}
\begin{table}[H]
\centering\
\small{
\begin{tabular}{l|c c c|c}
\textbf{Object flow} &\textbf{$INT_{i}$} & \textbf{$INT_{f}$}  & \textbf{$INT_{g}$}& \makecell{Net flow\\ on Cycle}\\
\hline
\g - flow   &   & & & 0  \\
\m - flow   & 4.24 - out &  2.82 - out & 7.06 - in & 0 \\              
\end{tabular}
}
\caption{Net flows on Cycle 1}
\label{tab:Net flows on Cycle 1}
\end{table}

\textbf{Nonperformance} When a node fails to send its required object, the affected chain or cycle is decomposed into a set of chains that place the nonperforming node in a bilateral contract with the neighbor node with whom it was obligated to send the object in their \ic. Below is the decomposition of Chain 7 induced by a failure of $INT_{i}$ to send its net obligations of 4.24 units of \m. The flows of \g and \m along edges are unchanged. Overall, agent profit and gross obligations are unchanged.

Chain 7a: 
$UL_{g}$
\tikz[baseline]{
  \draw[->] (0,0.1) -- (1,0.1) node[midway, above] {2\g};
  \draw[<-] (0,-0.1) -- (1,-0.1) node[midway, below] {13.06\m};
}
$INT_{f}$
Chain 7b: 
$INT_{i}$
\tikz[baseline]{
  \draw[->] (0,0.1) -- (1,0.1) node[midway, above] {2\g};
  \draw[<-] (0,-0.1) -- (1,-0.1) node[midway, below] {6.0\m};
}
$INT_{g}$
$INT_{g}$
\tikz[baseline]{
  \draw[->] (0,0.1) -- (1,0.1) node[midway, above] {2\g};
  \draw[<-] (0,-0.1) -- (1,-0.1) node[midway, below] {13.06\m};
}
$UB_{f}$

$\underset{\text{Recovered initial trade}}{\underbrace{\text{Chain 7c:} 
INT_{f}
\tikz[baseline]{
  \draw[->] (0,0.1) -- (1,0.1) node[midway, above] {2T};
  \draw[<-] (0,-0.1) -- (1,-0.1) node[midway, below] {10.24M};
}
INT_{i} 
}}$

The notable feature is Chain 7c, where agents $f$ and $i$ are placed in a bilateral contract. In the new contract $i$ is obligated to send 10.24 units of \m
in exchange for 2 units of \g . The unit price of \$5.12 is the initial contract terms between $f$ and $i$, scaled to the 2 units of \g assigned to Chain 7c (Table \ref{tab:Initial Contracts}). This reflects a property of RepoMech. When an agent fails to perform  - agent $g$ in this example - it is placed in a bilateral contract with the initial contract counterparty to whom it was obligated to send the object - agent $i$ in this example. This result reflects that counterparty risk is unaffected by the rearrangement of trades.

\textbf{Margin payments} When the \ic represents a forward contract, contracting parties are subject to margin requirements which are periodic additions or subtractions from the escrow accounts of agents associated with the end-nodes. The margin obligations are typically based on changes in the value at risk (``VAR''). Denoting \g as a security and \m as money for this purpose, VAR is composed of the changes in market price of \g and market volatility from the immediate prior margin adjustment. When VAR increases, the net escrow requirement 
increases for the seller (and correspondingly decreases for the buyer). The opposite occurs when VAR decreases. Subject to uniform margining formulas, margin escrow on a chain is only paid by end-nodes, since the net flow of \g is zero for all intermediate nodes. There is no margin requirement on cycles, since the flow of \g is netted to zero for every node.

\section{Accounting and Legal Features of RepoMech}
\label{sec:Accounting and Legal Features of RepoMech}

In this section we present an overview of the accounting treatment of RepoMech and its relation to bankruptcy remoteness. Section \ref{subsec: Accounting treatment of RepoMech} explains the accounting treatment and balance-sheet impact of RepoMech.  Section \ref{subsec:Bankruptcy safe-harbor of the replacement second-leg contracts} verifies that RepoMech maintains the repo bankruptcy safe harbor. Finally, Section \ref{subsec:Replacement first-leg contracts} shows how the termination and replacement of first-leg contracts can be incorporated into RepoMech. Appendix \ref{app:Repo Accounting Rules} provides a detailed discussion of repo accounting rules.

\subsection{Accounting treatment of RepoMech}
\label{subsec: Accounting treatment of RepoMech}

Each chain formed by the decomposition of the TFN has a fixed volume of \g assigned to each edge with an $UL$ node on the left end, an $UB$ node on the right end and $INT$ nodes in the middle. The replacement contract nets trades on each chain and cycle. The result is that all matched-trade $INT$ volume is fully netted in \g.  The accounting implications is that all matched trades are treated as first-leg final sales.\footnote{A necessary caveat throughout this section is that the outcome depends on application of the remaining FASB Topic 860 conditions and, for the linked second-leg position, separate FASB Topic 815 analysis to the specific transaction.} The un-netted $UL$ and $UB$ trades may, or may not, be treated as first-leg final sales.

\subsubsection{Intermediate nodes on chains and  nodes on cycles}

The key fact concerning the accounting treatment of the trades of intermediate nodes on chains and nodes on cycles is that the transfer of the security at the second-leg is netted out, so that the node does not send or receive the security (except in the event of a nonperformance by another node). The absence of a repurchase transaction implies that the first-leg of trades assigned to intermediate nodes on chains and cycles are treated as final sales. The second-leg is the FMV of the money payment the node is scheduled to receive under the replacement second-leg contract (which could include estimated cost related to nonperformance of nodes). For node $i$ on the repo chain in Figure \ref{fig:First and Second-Leg Repo Chain}
the upper bound of FMV is $T(p^{2}_{j\to i} - p^{2}_{i \to h})$. Figure \ref{fig:First-Leg Balance-Sheet Impact for Intermediate Nodes} depicts the first-leg balance-sheet impact for an intermediate node. 

\begin{figure}[H]

\begin{center}
\text{Node $i$ on Figure 1 repo chain}\vskip5pt
\begin{tabular}{l|l}
\hline
\multicolumn{1}{c|}{\textbf{Liabilities}} & \multicolumn{1}{c}{\textbf{Assets}} \\
\hline
& \underline{If sale accounting applies to the first leg}\\
& $Tp^{1}_{h\to i} - Tp^{1}_{i\to j}$\\
& \underline{Linked second-leg position}\\
& $\pm\,\mathrm{FMV}(\text{replacement second-leg commitment})$\\
\hline
& $\Delta A \approx$ net intermediation margin $\pm$ FMV\\
\end{tabular}
\end{center}
\caption{Illustrative balance-sheet impact for intermediate nodes under RepoMech if sale accounting applies}
\label{fig:First-Leg Balance-Sheet Impact for Intermediate Nodes}
\end{figure}

A comparison of Figure \ref{fig:First-Leg Balance-Sheet Impact for Intermediate Nodes} to the pre-reform accounting treatment of repo intermediaries in Figure \ref{fig:First-Leg Balance-Sheet Impact of Repo Intermediation: Pre-Reform} and the post-reform accounting treatment in Figure \ref{fig:First-Leg Balance-Sheet Impact of Repo Intermediation: Post-Reform}, shows that RepoMech reflects the balance-sheet impact of pre-reform accounting, which is applications is case-specific.

\subsubsection{End-nodes}
\label{subsubsec:End-nodes}

The accounting treatment of end-nodes depends on whether the trades are classified as repo trades or first-leg final sales with an embedded derivative second-leg. Repo accounting rules do not directly address the replacement second-leg contract structure, due to its unprecedented uniqueness. The crucial issue concerns the fact that under the \rtm the agent to whom the initial owner sends the \fa at the first-leg is different from the agent from whom the initial owner receives back the \fa at the second-leg. The accounting rules do not state whether this would disqualify the trade from being treated as a secure financing. The alternative would be a final sale accounting. In Appendix \ref{subsubsec:End-nodes of chains} we address this question in more detail.

When it is desired to ensure that end-node trades are treated as secured financings -- in order to prevent an agent from removing a net first-leg sale of the security from its balance-sheet -- the continuity of counterparties can be established by a procedure that mirrors the treatment of second-leg contracts. First, before trade occurs, assign initial first-leg trades to the same chains and cycles to which the associated initial second-leg trades are assigned (Section \ref{subsec:Replacement first-leg contracts}). Second, terminate initial first-leg contracts and replace with contracts on chains where end-nodes are counterparties and intermediate nodes are guarantors, matched with corresponding replacement second-leg contracts with end-node counterparties. By undertaking these steps the end-nodes retain the same counterparty and financial object volume at the first and second-leg. This supports the case to make these trades eligible to be treated as a secured financed repo trade.  

\subsubsection{Supplementary leverage ratio representations}

An initial repo trade increases the SLR denominator by $\triangle A$. For intermediate nodes on chains and node on cycles $\triangle A$ is the intermediation margin, which is a fraction of $\triangle A$ for post-reform repo intermediation (Figure \ref{fig:First-Leg Balance-Sheet Impact of Repo Intermediation: Post-Reform}), which is the \fl sale price.  The $UL$ initial repo lender end-node do not affect the SLR, which is also the case for the $UL$ trades under post-reform secured-financing accounting. The affect the $UB$ initial repo borrower end-nodes have on the SLR depends on their accounting treatment. Under final-sale derivative accounting the impact is the total intermediation margin. Under secured-lending accounting it is the \fl sale price of the \fa.

\subsection{Bankruptcy safe-harbor of the replacement second-leg contracts}
\label{subsec:Bankruptcy safe-harbor of the replacement second-leg contracts}

A feature of repo contracts that is valued by repo lenders such as money market funds is that, in the event of a default by a bankrupt borrower, the lender can immediately sell the security to a third party in order to realize all, or a portion of the money that it is owed at the second-leg. This right cannot be nullified by a Court.\footnote{There is a literature on the welfare implications of the repo safe-harbor. For a discussion of the issues see e.g. Duffie and Skeel \cite{DuffieSkeel2013}. We do not contribute to this debate. Our purpose is to demonstrate thatRepoMech does not alter this aspect of the legal environment for repo trading.} We show that the bankruptcy safe-harbor applies to replacement second-leg contracts underRepoMech.

\textbf{Safe-harbor for repo contracts} When a counterparty to a repo trade files for Bankruptcy under Chapter 11 of the U.S. Code, the trade cannot be avoided by the Court (11 U.S. Code § 546(f), 11 U.S. Code § 362(b)(7), and 11 U.S. Code § 561). The non-bankrupt seller has the statutory right to offset its loss by liquidating the security and the non-bankrupt purchaser has the statutory right to offset its loss by acquiring the security. 

\begin{quote}The exercise of a contractual right of a repo participant or financial participant to cause the liquidation, termination, or acceleration of a repurchase agreement...shall not be stayed, avoided, or otherwise limited by operation of any provision of this title or by order of a court or administrative agency in any proceeding under this title. (11 U.S. Code § 559)
\end{quote}

\begin{quote}
(b) The filing of a petition...does not operate as a stay...(b)(7) of the exercise by a repo participant or financial participant of any contractual right...under any security agreement or arrangement or other credit enhancement forming a part of or related to any repurchase agreement, or of any contractual right to offset or net out any termination value, payment amount, or other transfer obligation arising under or in connection with 1 or more such agreements, including any master agreement for such agreements.   (11 U.S. Code § 362)
\end{quote}

\textbf{Safe-harbor for replacement second-leg contracts}RepoMech is a multilateral netting contract between all participants. The authority to re-arrange second-leg trades in accordance with the protocol is stated in each initial repo trade contract.RepoMech  implements a multilateral netting of second-leg repo trades whereby each agent sends and receives the same net volume of money and the security as the initial repo contracts. The replacement second-leg contracts, which are the netting agreements associated with the initial repurchase agreements, are granted safe-harbor under multiple sections of the bankruptcy code.

\begin{quote}
(f)...the trustee may not avoid a transfer made by or to (or for the benefit of) a repo participant or financial participant, in connection with a repurchase agreement... (j) the trustee may not avoid a transfer made by or to (or for the benefit of) a master netting agreement participant under or in connection with any master netting agreement or any individual contract covered thereby...except to the extent that the trustee could otherwise avoid such a transfer made under an individual contract covered by such master netting agreement. (11 U.S. Code § 546) 
\end{quote}

\begin{quote}
(b) The filing of a petition...does not operate as a stay...(b)(27)of the exercise by a master netting agreement participant of any contractual right...under any security agreement or arrangement or other credit enhancement forming a part of or related to any master netting agreement, or of any contractual right...to offset or net out any termination value, payment amount, or other transfer obligation arising under or in connection with 1 or more such master netting agreements to the extent that such participant is eligible to exercise such rights... for each individual contract covered by the master netting agreement...(11 U.S. Code § 362)
\end{quote}

The conclusion is that agents who participate inRepoMech retain the safe-harbor protections that apply to the initial repo trades.

\subsection{Replacement first-leg contracts}
\label{subsec:Replacement first-leg contracts}

The trading mechanism can be expanded to implement an assignment of first-leg trades to chains and cycles. One possible reason for doing so would be to reduce the volume of financial objects flowing between agents at the \fl. Initial \fl contracts can be terminated and replaced by replacement \fl contracts on the chains and cycles using the same method as discussed in Section \ref{sec:RepoMech}. The direction of flow of money and the security is reversed. The \fl  and \sel volume of \g is the same, with the direction of flow reversed. This implies our method will generate the same set of chains and cycles as was generated from \sel flows of \g. The \fl prices may differ from the \sel prices. Below we display the \fl of Chain 7. for continuity with our presentation above, the second-leg identities of $UL$ and $UB$ are fixed, and their flows are reversed at the first-leg (i.e. $UL$ receives \g at the first-leg and $UB$ sends \g at the first-leg).

{\small{
Chain 7: $UL_{g}$ 
\tikz[baseline]{
  \draw[<-] (0,0.1) -- (1,0.1) node[midway, above] {2\g};
  \draw[->] (0,-0.1) -- (1,-0.1) node[midway, below] {\$12.44};} $INT_{f}$ 
\tikz[baseline]{
  \draw[<-] (0,0.1) -- (1,0.1) node[midway, above] {2\g};
  \draw[->] (0,-0.1) -- (1,-0.1) node[midway, below] {\$9.20};} $INT_{i}$ 
\tikz[baseline]{
  \draw[<-] (0,0.1) -- (1,0.1) node[midway, above] {2\g};
  \draw[->] (0,-0.1) -- (1,-0.1) node[midway, below] {\$10.80};} $INT_{g}$ 
\tikz[baseline]{
  \draw[<-] (0,0.1) -- (1,0.1) node[midway, above] {2\g};
  \draw[->] (0,-0.1) -- (1,-0.1) node[midway, below] {\$12.44};} $UB_{f}$
}}

Because the replacement first-leg contracts are reflections of the replacement second-leg contracts, they satisfy the replacement contract invariance stated in Proposition 1. A node's failure to send the required financial object will initiate the same recursive decomposition of chains and cycles, termination and issuance of new contracts, as described in Section \ref{sec:RepoMech}. The accounting treatment of trades with first and second-leg assigned to intermediate node on chains and nodes on cycles would not change as a result of the reorganization of first-leg trading. Those trades are netted for \g on both legs. Consequently, the balance-sheet impact of \fl trading for those nodes is their profit margin. This is the same as the impact of final sales.  The accounting treatment of trades that are assigned to end-nodes of chains remains open as to whether they are treated as first-leg final sales or secured financings (Section \ref{subsec: Accounting treatment of RepoMech}).

\section{Comparing RepoMech to Central Clearing}
\label{sec:Comparing RepoMech to Central Clearing}

\begin{figure}[H]
\begin{center}
\includegraphics[page=1,width=346pt,height=250pt,keepaspectratio]{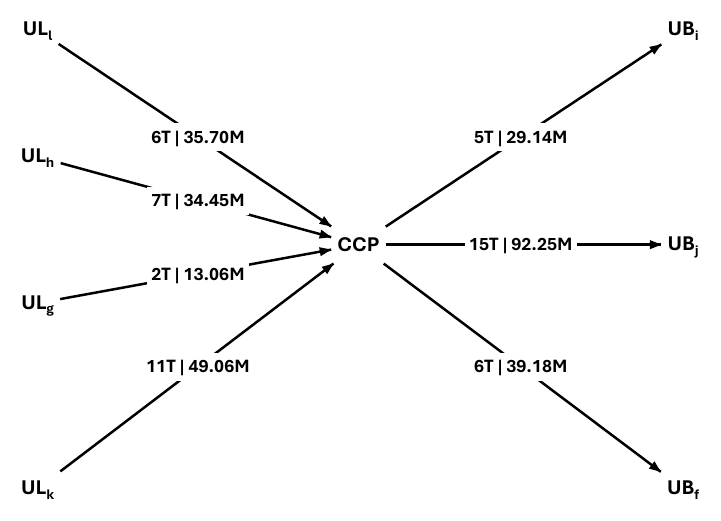}
\end{center}
\caption{Central Clearing}
\label{fig:Central Clearing}
\end{figure}

There is a straightforward relationship between our trading mechanisms and central clearing. Both can be derived from the TFN (Figure \ref{fig:Trade Flow Network}). The TFN allocates the un-netted trades to the net sending \g nodes, $UL$, located on the left side, and the net receiving \g nodes, $UB$, located on the right. In our trading mechanism the TFN is decomposed into chains and cycles. These nodes, and the flows along the edges attached to them, become end-nodes of chains. The $INT$ nodes are in the middle of chains and on cycles. Their \g-flows  are fully netted. 

Figure \ref{fig:Central Clearing} displays the transformation of the TFN to central clearing, with the flow of \m re-attached. The $UL$ and $UB$ nodes are connected to the CCP and the trades between $INT$ nodes are extinguished. This is the multilateral netting of central clearing, with agents paying to, or receiving from, the CCP the net \m flows of their $INT$ nodes. The process by which this occurs is that the trade contracts entered into by agents are novated and replaced by contracts with the CCP.

\subsection{Accounting treatment of central clearing}
\label{subsec:Accounting treatment of central clearing}

Under the current institutional protocol for central clearing of Treasuries repo, initial first-leg contracts are novated and replaced by identical contracts with a central clearing counterparty ("CCP").\footnote{See Treasury Market Practices Group White Paper on Clearing and Settlement in the Market for U.S. Treasury Secured Financing Transactions \cite{TMPG-repo} for a detailed description of the central clearing of DVP repo.} The replacement first-leg contracts are executed on a trade-by-trade basis. An agent's second-leg flows of the security are partitioned into un-netted flow and matched-trade (netted) trades following the logic of the Parent to Children Node Splitting Algorithm. The un-netted flow trades are aggregated into a single transaction with the CCP and are treated as a secured financing.\footnote{The aggregate second-leg un-netted trade flows of money and the security can be allocated to match first-leg contracts for accounting purposes.} Under current practice the  remaining matched-trades are usually treated as first-leg final sales with second-leg trades netted out and the agent sending or receiving the sum of the money flows from the CCP, which is recorded as a balance-sheet asset at FMV.\footnote{See Appendix \ref{subsec:Accounting treatment of centrally cleared repo} for the application of accounting rules to central clearing.}

\subsection{Balance-sheet impacts of RepoMech and central clearing}
\label{subsec:Balance-sheet impacts of RepoMech and central clearing}

Proposition 2 shows that the balance-sheet impact of repo trades under RepoMech is weakly smaller than the balance-sheet impact under central clearing.\footnote{Assuming that netted centrally cleared trades are treated as \fl final sales.}

\begin{proposition****}
Conditional on sale accounting being available for matched RepoMech trades, for any set of initial repo trades the increase in an agent's first-leg assets is weakly lower under RepoMech compared to central clearing, up to the second-leg FMV.
\end{proposition****}

\begin{proof}
Section \ref{subsec: Accounting treatment of RepoMech} showed that (i) the first-leg of all matched initial repo trades, which are equivalent to multilateral netted trades, are treated as final sales with the FMV of the second-leg included as an asset on the agent's balance-sheet and (ii) the remaining un-netted trades may be treated in the same way as matched-trades or or may be treated as secured financing, which has a larger balance-sheet impact. 

Section \ref{subsec:Accounting treatment of central clearing} showed that all matched initial repo trades have the same balance-sheet impact under RepoMech and central clearing and that all un-netted trades are treated as secured financings. 

The possibility of treating the first-leg of un-netted repo trades under RepoMech as final sales proves that the balance-sheet impact of RepoMech is weakly less than central clearing for a common FMV. This proves the proposition.  
\end{proof}

\subsection{Compatibility with mandatory central clearing of U.S. Treasuries repo}
\label{subsec: Compatibility with mandatory central clearing of U.S. Treasuries repo}

In December 2023 the U.S. Securities and Exchange Commission ("SEC") issued a final rule that mandates members of covered clearing agencies to centrally clear all transactions involving certain types of Treasuries trades (the "rule") \citep{SEC-UST-CCP-Rule-2023}. Repo trades are among the trades covered by the rule and are required to comply by June 30, 2027. The Fixed Income Clearing Corporation ("FICC") is currently the only covered clearing agency ("CCA") for Treasuries repo. Intermediaries carrying out the majority of repo volume are members of the FICC Government Securities Division ("GSD")  and currently clear their inter-dealer trades with FICC acting as the CCP \citep{Kahn2021}.\footnote{FICC GSD members are authorized to clear Treasuries trades on the FICC platform \url{https://www.dtcc.com/client-center/ficc-gov-directories}} Starting in 2027 these intermediaries must centrally clear all repo trades with non-FICC platform members. In the current environment where repo intermediation is dominated by regulated agents, most Treasuries repo transactions, as defined by the SEC, will be required to centrally clear through FICC. RepoMech is not a transaction type currently covered by the SEC rule. The rule states;

\begin{quote}
The proposed definition of an eligible secondary market transaction would include,
among other things, all U.S. Treasury repurchase and reverse repurchase agreements entered into by a direct participant of a U.S. Treasury securities CCA...in a U.S. Treasury repo transaction, one party sells a U.S. Treasury security to another party ... and commits to repurchase the security at a specified price on a specified later date... and a reverse repo transaction is the same transaction from the buyer’s perspective. (SEC  \cite{SEC-UST-CCP-Rule-2023} Section II. 2(a)) 
\end{quote}

The rule appears not to apply to trades assigned to intermediate nodes on chains or cycles, where there is no repurchase transaction, because those assignments do not fit the conventional same-counterparty repurchase structure described in the rule text and in FICC definitions. Additionally, it may not apply to end-node transactions (Section \ref{subsubsec:End-nodes of chains}). Further clarification is provided by a footnote that references the FICC definition of the repo trade which requires that the same counterparties transact at the first and second-legs. 

\begin{quote}
(1) an agreement of a party to transfer Eligible Securities to another party in exchange for the receipt of cash, and the simultaneous agreement of the former party to later take back the sale Eligible Securities...\textbf{from the latter party} in exchange for payment of cash...(SEC \cite{SEC-UST-CCP-Rule-2023} Section II. 2(a) viii.)
\end{quote}

This appears to rule out application to trades assigned to the end-nodes of chains, where the counterparties - in terms of the principal senders of money and the security - change at the second-leg. Finally, the rule applies only to transaction types that are currently cleared by FICC.

\begin{quote}
The definition of an eligible secondary market transaction, both as proposed and as
adopted, applies to all types of transactions that are of a type currently accepted for clearing at a U.S. Treasury securities CCA. It does not impose a requirement on a U.S. Treasury securities CCA to offer additional products for clearing. (SEC  \cite{SEC-UST-CCP-Rule-2023} Section II. 2(a) viii.)
\end{quote}

FICC does not currently offer a program to clear multiparty contracts and therefore cannot clear transactions in RepoMech. The conclusion is that RepoMech is not affected by, and therefore is compatible with, the SEC's central clearing mandate for U.S. Treasuries repo trades. An implication of the compatibility is that an FICC GSD clearing member can operate under both regimes at the same time. It can centrally clear some of its repo trades while carrying out other trades under RepoMech.

\subsection{Profit, counterparty risk and cost}
\label{subsec: Profit, counterparty risk and cost}

Neither RepoMech or central clearing alter an agent's contractual profit. Where the two regimes differ are on counterparty risk and cost. Under RepoMech an agent's counterparty risk remains with its initial repo contract counterparties ( Proposition 1(2)). Under central clearing counterparty risk shifts to the CCP, which is the counterparty to all trades. The difference in counterparty risk has implications for incentives to enter into initial repo contracts and market risk which we do not explore here. 

Divergent costs can arise from the cost of operating the protocols and differences in institutional practice. RepoMech has not yet been implemented, but it is noteworthy that the underlying protocol is not computationally intensive (Theorem 1). FICC charges fees, imposes margin payments and maintains a default fund on its central clearing platform. A crucial difference between RepoMech and central clearing is that the latter injects a third party, the CCP, into the network of trades. Both the SEC and FICC project that the rule will increase trading costs in the Treasuries repo market. The SEC acknowledges that "additional clearing likely would result in increased margin contributions and clearing fees" (SEC  \cite{SEC-UST-CCP-Rule-2023} Section II.2(b)i). Reporting on the results of a survey of market participants it conducted, FICC "expects that the incremental indirect participant Treasury volume could result in a corresponding increase in Value at Risk (VaR) margin, which it conservatively estimates could be approximately \$26.6 billion across the FICC / GSD membership" (FICC \cite{DTCC2023}). Bank of New York Mellon, a major repo market participant, opines that the true increase in margin "could be higher because [the FICC survey] is based on a subset of market participants." and concludes.

\begin{quote}In repo markets, netting should reduce some balance sheet costs, but additional margin and the costs of sponsorship are likely to push bid-ask spreads wider, increasing the cost of repo funding and leverage. On net, in both Treasury cash and financing markets, liquidity in normal times is likely to be less continuous than many have come to expect. \citep{Wuerffel2024}
\end{quote}

In addition to the direct cost of increased margin payments, FICC central clearing prevents the use of margin posted for repo trades to be offset by margin on other types of trades between counterparties,  which is an indirect cost.\footnote{See Hempel et.al. \cite{Hempel2023} for a discussion of the common practice whereby hedge funds offset margin on Treasuries repo trades with margin on Treasuries futures trades with the same counterparty.}

\section{Conclusion}
\label{sec:conclusion-2}

In this paper we described how changes to repo accounting rules and bank leverage regulations made after the 2008 financial crisis interacted to reduce the volume of U.S. Treasuries repo that banks could intermediate. The updates to FASB accounting rules changed the increase in asset value recorded at the first-leg from total intermediation margin to first-leg sale price, which represented an order-of-magnitude increase. At the same time, the SLR reduced the volume of assets that a bank could record for a given amount of capital. The combination of these two changes lowered intermediation capacity in the U.S. Treasuries repo market and has been linked to strains in market liquidity.

We then adapted the multilateral netting mechanism in Aronoff et al.\ (2025) to the second-leg of a repo trade. We showed how RepoMech can be used to reduce the balance-sheet impact of repo intermediation by lowering the increase in recorded asset value at the first-leg, thereby lowering the impact on the SLR denominator. The mechanism does this by assigning second-leg trades to chains and cycles, terminating the initial contracts, and replacing them with multiparty contracts that net matched trades while preserving the underlying flows between initial counterparties.

Next, we compared RepoMech to central clearing. We showed that the balance-sheet reduction achieved by RepoMech is at least as large as that achieved by central clearing. The salient difference is that RepoMech does not alter the location of counterparty risk of the initial repo counterparties, whereas central clearing concentrates that risk on the CCP. RepoMech therefore captures the balance-sheet advantages of multilateral netting without requiring a new counterparty or a reallocation of bilateral credit exposure. Finally, we showed that RepoMech preserves contractual profit, is bankruptcy remote in the same sense as DVP repo, and is compatible with the SEC's central clearing mandate for U.S. Treasuries repo trades.

\small{\bibliographystyle{ACM-Reference-Format}
\bibliography{references}
}

\newpage
\appendix

\section{Repo Accounting Rules}
\label{app:Repo Accounting Rules}

In this appendix we discuss in detail the application of accounting rules to repo trades. We refer to the accounting standards of the Financial Accounting Standards Board ("FASB"), which are numerically organized in accord with the Accounting Standards Codification ("ASC") \citep{FASBASC} and updated by Accounting Standards Updates ("ASU"s)  \citep{FASBASU}. FASB standards are the primary source of Generally Accepted Accounting Principles ("GAAP"), which are the corporate accounting standards under U.S. law.\footnote{GAAP incorporates standard from other sources including the American Institute of Certified Public Accountants ("AICPA") and the Securities and Exchange Commission ("SEC").} We also refer to the commentary on FASB standards in the PWC Viewpoint online book \cite{PwCViewpoint}. 

Section \ref{subsec:Accounting rules for repo trades} explains the FASB criteria for classifying the first-leg of a repo trade as a final sale or a secured financing, the balance-sheet treatment under each regime and the ASUs issued after the 2008 financial crisis that substantially narrowed sale accounting for repo transactions and specifically required repurchase-to-maturity transactions to be treated as secured borrowings.

Section \ref{subsec:Accounting treatment of RepoMech} explains why the replacement second-leg contracts of intermediate nodes on chains should be classified as derivative transactions and how that affects an agent's balance-sheet. Finally, Section \ref{subsec:Accounting treatment of centrally cleared repo} describes the accounting treatment of centrally cleared repo and compares it to the accounting treatment of the replacement second-leg contracts.   

\subsection{Accounting for repo trades}
\label{subsec:Accounting rules for repo trades}

\subsubsection{FASB rules}
\label{subsubsec:FASB rules}

The representation of a repo trade on an agent's balance sheet depends on whether the first-leg transaction is classified as a final sale or a secured borrowing. ASC 860-10-40-5 is a list of conditions that must be met in order for a transfer of a \scu to be classified as a final sale. For repos, the effective-control condition in ASC 860-10-40-5(c)(1) is often the decisive issue, although sale accounting still requires the isolation and transferee-rights conditions in ASC 860-10-40-5(a) and (b) to be satisfied as well.

\begin{quote}
(c) Effective control. The transferor [does] not maintain effective control over the transferred \scus. A transferor's effective control over the transferred [\scus] includes...

\begin{adjustwidth}{2em}{0em}
(1) An agreement that both entitles and obligates the transferor to repurchase or redeem the transferred [\scu] 
\end{adjustwidth}
\end{quote}

ASC 860-10-40-24 lists the conditions that must be met to satisfy ASC 860-10-40-5 (c)(1). Prior to the post 2008  financial crisis reforms, the key provision was (b), which stated; 

\begin{quote}
The transferor is able to repurchase or redeem [the transferred \scus] on substantially the agreed terms, even in the event of default by the transferee. To be able to repurchase or redeem \scus on substantially the agreed terms, even in the event of a default by the transferee, a transferor must at all times during the contract term have obtained cash or other collateral [or ``\scu''] sufficient to fund substantially all of the cost of purchasing replacement \scus from others...
\end{quote}

The first-leg of a repo trade that failed to satisfy ASC 860-10-40-24(b) was a final sale. This provision is violated if the repurchase price of the \scu is below its market price. In the event of a default by the transferee, the transferor's repurchase price would be the market price, which exceeds the contract price. The transferee would not be able to repurchase the \scu  "on substantially the agreed terms". In reaction to the controversial role of Repo 105 in the 2008 financial crisis, in 2011 ASC 860-10-40-24(b) was repealed by ASU 2011-3. Under the superseded implementation guidance, arrangements with as much as 98 percent collateralization generally fell within the guideline. Repo 105 was notable because it was structured to fall outside that ordinary range (see Appendix \ref{subsec:Lehman's Repo 105}).

Repeal of ASC 860-10-40-24(b) removed one important route by which a standard repo could satisfy the effective-control analysis, but it did not by itself eliminate sale accounting for every repo-like transfer. Before ASU 2014-11, repurchase-to-maturity transactions were generally analyzed under ASC 860-10-40-5 and were commonly accounted for as sales with forward repurchase commitments when those derecognition conditions were met. MF Global used this exception to remove risk exposure to low-rated sovereign debt from its balance-sheet. In 2014 FASB mandated that repo-to-maturity trades are accounted as secured borrowing; ``A repurchase-to-maturity transaction shall be accounted for as a secured borrowing as if the transferor maintains effective control.'' (ASU 2014-11)

\subsubsection{Balance-sheet impact on repo intermediaries}
\label{subsubsec:Balance-sheet impact on repo intermediaries}

We now compare the first-leg balance-sheet impact of repo intermediation under final-sale and secured-financing accounting, with reference to the $INT$ nodes in Figure \ref{fig:First and Second-Leg Repo Chain}. \footnote{Salerno et.al. \cite{Salerno-Repo-Accounting} provides a balance-sheet example of the difference between final-sale and secured financing accounting for repo.}

\textbf{Final-sale accounting} The initial owner of $T$, $UB$, subtracts the carrying value of $T$ and adds the first-leg sale proceeds to its assets. Each $INT$ intermediary subtracts the first-leg purchase price and adds the first-leg sale price of $T$ to its assets.\footnote{$INT_{i}$ adds the value of $T$ from its purchase and subtracts it from its sale, which cancels out.} $UL$ adds the value of $T$ and subtracts the first-leg purchase price to its assets. The second-leg purchase contract (with its left neighbor) and sale contract (with its right neighbor) are forward contracts, which are treated as a derivative and recorded at fair-market-value ("FMV"), which can be positive or negative in accordance with ASC 815 (PWC Viewpoint: Derivatives and Hedging Guide Chapter 1.2.2  \cite{PWC-derivatives}). The calculation of FMV is the expected gap between the value of $T$ and its price. When repurchase price, $p^{2}$, is below market, $p_{T}$, the gap is positive and is recorded as an asset. When repurchase price is above market the gap is negative and is recorded as a liability.\footnote{FMV is adjusted by counterparty credit risk and value-at-risk (ASC 820 and 815)}. Figure \ref{fig:Balance-Sheet Impact of Repo Intermediation: Pre-Reform} displays the pre-reform \bs for intermediary $INT_{i}$ on the repo chain in Figure \ref{fig:First and Second-Leg Repo Chain}.\footnote{WLOG we assume that $INT_{i}$ carries $T$ on its \bs at the market price $p_{T}$.} For convenience we depict the FMV of the \sel purchase and sale as positive and therefore displayed as assets.

\begin{figure}[H]

\begin{center}
\text{$INT_i$ intermediation of $T$ on the repo chain in Figure \ref{fig:First and Second-Leg Repo Chain}}\vskip5pt
\begin{tabular}{l|l}
\hline
\multicolumn{1}{c|}{\textbf{Liabilities}} & \multicolumn{1}{c}{\textbf{Assets}} \\
\hline
& \underline{First-leg net cash margin}\\
& $+\,Tp^{1}_{h\to i} - Tp^{1}_{i\to j}$\\
& \underline{Second-leg commitments}\\
& $\pm\,\mathrm{FMV}(\text{purchase from }h)$\\
& $\pm\,\mathrm{FMV}(\text{sale to }j)$\\
\hline
& $\Delta A \approx$ net intermediation margin $\pm$ net FMV\\
\end{tabular}
\end{center}
\caption{Balance-sheet impact at the first-leg date under pre-reform final-sale accounting}
\label{fig:Balance-Sheet Impact of Repo Intermediation: Pre-Reform}
\end{figure}

A key observation regarding balance-sheet impacts of final-sale repo accounting is that the \scu $T$ is recorded on the balance-sheet where it is owned. Prior to the trade it is on the balance-sheet of the initial owner $UB$. After the first-leg trade it is on the balance-sheet of its interim owner $UL$. Between the first and second-legs $T$ does not appear on the balance-sheet of its initial owner $UB$. This can be derived from Figure \ref{fig:Balance-Sheet Impact of Repo Intermediation: Pre-Reform} by eliminating the first-leg purchase and the second-leg sale. Finally, the $INT$ intermediaries, who purchase and then sell $T$, record only the intermediation margin on their balance-sheet because they are not the owners of the \scu at completion of the first-leg along the repo chain.

\textbf{Secured-financing accounting} ASC 860-30-25-2 states that a transfer of \scus that does not meet the conditions of a sale under ASC 860-10-40-5 should be accounted for as a secured borrowing. At a high level the balance-sheet impact of a repo trade conforms to the following (PwC Viewpoint: Transfers and Servicing of [\scus] guide Chapter 5.2  \cite{PWC-repo});

\begin{quote}
[U]nder the secured borrowing accounting model, the transferor:

\begin{adjustwidth}{2em}{0em}
\begin{itemize}
\item  Recognizes any cash received from the transferee (and any other [\scus] obtained from the transferee that the transferor can pledge or exchange, other than beneficial interests in the transferred [\scus])

\item Records an obligation (liability) to return the cash to the transferee (and any other recognized [\scus] obtained from the transferee)
\end{itemize}
\end{adjustwidth}

Under the secured borrowing accounting model, the transferee:

\begin{adjustwidth}{2em}{0em}
\begin{itemize}
\item Derecognizes any cash paid to the transferor

\item Records a receivable, representing its entitlement to receive at a later date the cash paid to the transferor

\item Does not record the [\scus] obtained from the transferor (barring a default by the transferor)
\end{itemize}
\end{adjustwidth}
\end{quote}

The implication is that the repo borrower recognizes the first-leg cash inflow and a corresponding secured-borrowing liability, while retaining the transferred security on balance sheet. The repo lender derecognizes the cash paid, records a receivable, and does not recognize the transferred security absent default by the borrower.\footnote{A gap between the \fl and \sel price is recorded at the \sel.} A repo intermediary is both borrower and lender, with the exception that it does not record the \fa on its \bs. The \fl \bs impact is the same as for a repo borrower. Notably, the value of the \scu is not removed from the balance-sheet of the initial owner, and therefore does not appear on any other balance-sheets along the repo chain. Figure \ref{fig:Balance-Sheet Impact of Repo Intermediation: Post-Reform} depicts the balance-sheet impact of intermediating repo under secured financing accounting, where the intermediary is both transferor (seller) and transferee (buyer) at each leg.  

\begin{figure}[H]

\begin{center}
\text{$INT_i$ intermediation of $T$ on the repo chain in Figure \ref{fig:First and Second-Leg Repo Chain}}\vskip5pt
\begin{tabular}{l|l}
\hline
\multicolumn{1}{c|}{\textbf{Liabilities}} & \multicolumn{1}{c}{\textbf{Assets}} \\
\hline
+ secured borrowing obligation & + cash received on borrowing leg \\
(initially = cash received; & $(= Tp^{1}_{h\to i})$ \\
accretes to repurchase amount) & \\
\hline
& $\Delta A \approx +\,Tp^{1}_{h\to i}$\\
\end{tabular}
\end{center}
\caption{Net first-leg balance-sheet impact of repo intermediation under secured-borrowing accounting}

\label{fig:Balance-Sheet Impact of Repo Intermediation: Post-Reform}
\end{figure}

\textbf{Comparison of final-sale and secured financing accounting} There are two salient differences in balance-sheet impacts between final-sale and secured financing accounting of repo trades. One is that the increase in recorded assets from intermediation is an order of magnitude higher under secured financing accounting. $\triangle A$ under final-sale accounting is the intermediation margin (Figure \ref{fig:Balance-Sheet Impact of Repo Intermediation: Pre-Reform}) versus the entire value of the \scu under secured financing accounting (Figure \ref{fig:Balance-Sheet Impact of Repo Intermediation: Post-Reform}). The other difference is that, under final-sale accounting, an initial owner can remove the \scu from its balance-sheet at the first-leg. Under secured financing accounting the initial owner cannot remove the \scu from its balance-sheet. It this feature that prevents an agent from using repo to disguise its ownership of select \scus.

One can take the view that secured-financing accounting does not represent the underlying economic substance of a repo trade for two reasons. One reason is that  ownership of the \scu is not recorded on the balance-sheet of its owner after the first-leg. The second reason is that de-coupling the second-leg repurchase obligation and sale price into separate liability and asset mis-represents the risk they embody. Specifically, the obligation to sell is contingent on the receipt of the \scu (that is the meaning of delivery-versus-payment). Therefore, the true risk is the gap between the two. The same argument applies in reverse to the sale transaction. Making these adjustments takes us back to final-sale accounting. This has an important implication for the interpretation of bank capital and leverage regulations. It implies that the recorded assets of financial intermediaries overstate the economic value of assets. 

On the other hand, secured-financing accounting prevents an agent from removing assets from it balance-sheet by acting as a repo borrower and making a first-leg sale. This is what enabled Lehman's Repo 105. An optimal solution would rule out this type of abuse without requiring \bs recording that mis-represents the true underlying risk of the repo trade. Chircop  \cite{Chircop2012}  addresses this dilemma and a proposes a modification to accounting rules that provide a solution.

\subsection{Accounting treatment of RepoMech}
\label{subsec:Accounting treatment of RepoMech}

RepoMech partitions an agent's initial repo trades into excess borrowing or lending and matched trades. The second-leg trades of the former are assigned to end-nodes of chains and the second-leg trades of the latter are assigned to intermediate nodes on chains or cycles. There is a fixed amount of the \scu assigned to each chain and cycle. Neighboring nodes on a chain or a cycle are counterparties under initial repo contracts and the flows of money and \scu between them are assigned from the flows in their initial repo contracts. The initial second-leg contracts are terminated and replaced by \mln second-leg contracts on each chain and cycle. A failure by a node to send the required financial object triggers a decomposition of the chain or cycle formed by pulling out the nonperforming node and the neighboring node to whom the unsent object flowed along the connecting edge; a 
termination of the \mln contract on the chain or cycle and its replacement by \mln contracts on the newly formed chains. This process occurs recursively for chains with more than two nodes until there is no nonperformance on any chain with more than two nodes. A chain with two nodes is a bilateral contract that cannot be further decomposed. A failure to perform under a bilateral contract is a final default. Each bilateral contract replicates the terms of the initial \sel contract between the counterparties, scaled to the volume of \fa assigned to the chain (Section \ref{sec:RepoMech}).

\subsubsection{Intermediate nodes on chains and nodes on cycles}

A key difference from an accounting standpoint between the \rtm and bilateral repo trades is the trades assigned to the middle of chains is that the \mln contract eliminates the obligation to repurchase the \fa or those trades. Additionally, since agents commit to the \rtm protocol before execution of their \fl trades, the netting is contractually binding at the \fl. For intermediate-node trades, the absence of a bilateral repurchase obligation may mean that effective control is no longer maintained. Sale accounting for the first leg therefore appears plausible, provided the other derecognition conditions in ASC 860-10-40-5 are also satisfied. Any resulting forward or linked commitment should then be analyzed separately under FASB Topic 815 to determine whether derivative accounting applies and, if so, whether the instrument is recognized at fair value.

An argument may proceed as follows: The replacement \sel contracts are forward contract to repurchase at a later date. It meets the definition of a derivative instrument under ASC 815-10-15: it has an underlying (the fixed prices/interest rates on the Treasury repo), a notional amount ($T$ units of the Treasury), requires no new initial net investment (the initial exchange of cash for the security has already occurred in the first leg), and will be settled by a net payment in cash (PWC Viewpoint: US Derivatives \& hedging guide - Chapter 4: Embedded derivative instruments \cite{PWC-derivatives}). The \sel of these trades are recorded on the \bs at FMV, i.e. the net value of the trade, adjusted by counterparty risk. If accepted, this would imply the same treatment as the pre-reform accounting and is displayed in Figure \ref{fig:Balance-Sheet Impact of Repo Intermediation: Pre-Reform}.

\begin{figure}[H]

\begin{center}
\text{Node $i$ on Figure 1 repo chain}\vskip5pt
\begin{tabular}{l|l}
\hline
\multicolumn{1}{c|}{\textbf{Liabilities}} & \multicolumn{1}{c}{\textbf{Assets}} \\
\hline
& \underline{If sale accounting applies to the first leg}\\
& $Tp^{1}_{h\to i} - Tp^{1}_{i\to j}$\\
& \underline{Linked second-leg position}\\
& $\pm\,\mathrm{FMV}(\text{replacement second-leg commitment})$\\
\hline
& $\Delta A \approx$ net intermediation margin $\pm$ FMV\\
\end{tabular}
\end{center}
\caption{Illustrative balance-sheet impact for intermediate nodes under RepoMech if sale accounting applies}
\label{fig:First-Leg Balance-Sheet Impact for Intermediate Nodes: RepoMech}
\end{figure}

Finally, it should be noted that the property of the \rtm protocol that preserves counterparty risk does not affect the applicability of ASC 815-10-15. The criteria for treatment as a secured borrowing in ASC 860-10-40-5(c)(1) requires that the following holds. 

\begin{quote}
An agreement that both entitles and obligates the transferor to repurchase or redeem the transferred [\scus].
\end{quote}

For \sel trades assigned to a node in the middle of a chain or on a cycle,  the agent is not entitled to repurchase the [\scu] under the replacement \sel contract, nor is it obligated to do so. It can only recover the repurchase right and obligation in the initial \sel upon the contingent occurrence of a nonperformance by the agent or its neighboring node. ASC 860-10-40-25 indicates that transfers containing only conditional rights or only conditional obligations generally may not maintain effective control, because the conditions might not occur.

\begin{quote}
Transfers that include only the right to reacquire, at the option of the transferor or upon certain conditions, or only the obligation to reacquire, at the option of the transferee or upon certain conditions, may not maintain the transferor's control, because the option might not be exercised or the conditions might not occur. Similarly, expectations of reacquiring the same securities without any contractual commitments (for example, as in wash sales) provide no control over the transferred securities.   
\end{quote}

\subsubsection{End-nodes of chains}
\label{subsubsec:End-nodes of chains}

Trades that are assigned to end-nodes of chains are associated with repurchase obligations under the replacement \sel contracts that match the terms of the initial \sel contracts in terms of the inflow and outflow of financial objects. The salient difference between the initial repo trade and the trade under RepoMech is that the identity of the second-leg counterparty is not the same. The \sel contract is a \mpc. The question arises whether this affects the accounting of the end-node trades. Neither ASC 860-10-40-24 nor ASC 860-10-40-24(b) - which defines "effective control" - make an explicit statement about the identity of the \sel counterparty. However,  ASC 860-10-40-5 does make reference to a transferee that, at least implicitly assumes continuity of the same counterparty in the first and second-legs. 

It does not appear that FASB has explicitly addressed the case where the identity of the counterparty changes (or where multiparty contracts exist).\footnote{The authors are unable to find any references in the PWC Viewpoint book or the FASB Accounting Interpretations, which address the application of accounting rules in concrete circumstances.} Current Topic 860 does not squarely address a multiparty repurchase structure in which the first-leg and second-leg counterparties differ. Because the authoritative examples and market-rule definitions repeatedly assume the same counterparty, the accounting outcome for end-node trades should be presented as unresolved and highly structure-specific.

\subsection{Accounting treatment of centrally cleared repo}
\label{subsec:Accounting treatment of centrally cleared repo}

Under central clearing an agent's initial  first-leg contracts are canceled and replaced by identical contracts with a central clearing counterparty ("CCP"). The replacement first-leg contracts are executed on a trade-by-trade basis. The second-leg flows of the \scu are partitioned into two groups; matched trades and excess (inflow or outflow) trades.\footnote{This is accomplished by the Parent to Child Node-Splitting Algorithm, albeit possibly with a different assignment criteria.} Each first-leg trade is assigned to the group its associated second-leg is assigned to. 

For matched-trades, the agent sends to or receives from the CCP its second-leg net money inflow, which is the same payoff as under RepoMech. Consequently, these trades meet the requirements of ASC 860-10-40-5 and their first-legs are treated as final sales with balance-sheet impact as depicted in Figure \ref{fig:Balance-Sheet Impact of Repo Intermediation: Pre-Reform} with the  FMV composed of the net money inflow adjusted by counterparty credit risk.  

Excess repo trades satisfy the effective control definition of ASC 860-10-40-24 (they are standard repo trades) which suggests that likely they would not satisfy the requirements for treatment as a final sale (ASC 860-10-40-5) and  are therefore treated as secured financings with balance-sheet impact as depicted in Figure \ref{fig:Balance-Sheet Impact of Repo Intermediation: Post-Reform}. 

The substitution of the CCP as counterparty to all agents restructures risk-bearing. Risk shifts from initial repo contract counterparties to the CCP and the CCP bears the risk of agent default. This is reflected in the counterparty credit risk on the balance-sheet. Agent net cash flow is the same as in the initial repo contracts, assuming no counterparty defaults in either case.

\section{Lehman and MF Global Repo Strategies}
\label{app: Lehman and MF Global Repo Strategies}

In this appendix we discuss the strategies employed by Lehman and MF Global to use repo transactions to conceal risk positions.

\subsection{Lehman's Repo 105}
\label{subsec:Lehman's Repo 105}

The investment bank Lehman Brothers devised a transaction structure, called Repo 105 (and a similar one called Repo 108), which it employed around financial disclosure dates for several years prior to its 2008 bankruptcy. The 105 (and 108) reference the percentage haircut (5\% and 8\%) on the first-leg sale price. This ensured that the second-leg repurchase price was below market. This enabled Lehman to take the position that the repo trade did not qualify for treatment as secured financing because it did not satisfy pre-reform ASC 860-10-40-24(b). Namely, the financial asset could not be purchased in the market at or below the contract second-leg repurchase price. The maneuver involved an arbitrage between US and UK accounting rules. Lehman executed the repo transaction in the UK through an offshore subsidiary. The first leg was treated as a final sale under UK law. The reporting of the final sale was consolidated up to the US holding company balance sheet. Lehman used its holdings of subprime mortgage securities as collateral and applied the proceeds of the first leg sale to pay down debt. This enabled Lehman to simultaneously conceal its subprime exposure and to understate its indebtedness. Between the first-leg and the second-leg, Lehman's subprime securities were owned by its counterparty. They were removed from Lehman's balance-sheet during that time interval without the repurchase obligation recorded as a debt. Meanwhile, Lehman received a payment of $M$ at the first leg, which it used to pay down lines of credit. As a consequence, by timing Repo 105 so that its quarterly release of financial information occurred between the first and second-leg, Lehman was able to report a balance-sheet with a lower volume of subprime securities holdings and less borrowing than would be the case after the second-leg (which it had a legal obligation to complete). It is estimated that, for several years prior to its bankruptcy, Lehman's use of Repo 105 enabled it to under-report its holdings of subprime securities and its debt by approximately \$50 billion.\footnote{See Chang et.al  \cite{Chang-Repo105} and Hartwell  \cite{Hartwell-repo-105} for a technical description of how Repo 105 worked, and Pounder \cite{Pounder2011} for an explanation of how Lehman used repo 105 (and 108) to alter the balance-sheet it presented to regulators and investors and the response of regulators.}

\subsection{MF Global's Repo-to-Maturity Program}

At the second-leg of a repo-to-maturity transaction, the repo lender collects the payoff from the collateral security issuer on the maturity date in lieu of the repo borrower repurchasing the collateral. The borrower is only obligated to make a payment in the event the security issuer defaults on its obligation to repay at maturity. Under the pre-reform FASB repo accounting rules, the first-leg transaction was treated as a final sale. The second-leg obligations were not recorded on the balance-sheet. For two years prior to its 2011 bankruptcy, MF Global entered into repo-to-maturity transactions to conceal from investors and regulators its exposure to low-rated sovereign debt. An example of the strategy worked as follows. MF Global draws on a line of credit to purchase a risky high yielding Greek sovereign bond for a price of $M$. The purchase price of the bond is below its payoff at maturity, which is $M + \xi$ (we normalize the interest rate to zero for simplicity). The discount on the purchase price is $\xi$. Shortly thereafter, MF Global enters into a repo transaction in which is sells the Greek bond in the first leg for a price of $M + \frac{1}{2}\xi$ and sets the second leg date to match the maturity date of the Greek bond. MF Global is able to sell the bond for a higher price than it paid because it protects its counterparty against default risk by agreeing to periodically pay into a margin account sufficient $M$ to cover the implied risk of loss reflected in the cost of credit default swaps linked to the Greek bonds. MF Global uses the first leg sale proceeds to pay down the line of credit it used to acquire the Greek bonds and retains a profit equal to a portion of the discount on its purchase, $\frac{1}{2}\xi$. The repo-to-maturity accounting exempts MF Global from recording the second leg transaction on its balance sheet. In addition, since the margining obligation is off-balance sheet, it does not get reported. The repo-to-maturity strategy unraveled in 2011 when the implied risk on Greek bonds increased, which required MF Global to make large margin payments. The need for cash ultimately resulted in the misappropriation of millions of dollars from customer accounts. Soon afterward the company filed for bankruptcy.\footnote{For description of MF Global's repo-to-maturity trades see Hartwell \cite{Hartwell-repo-105}}

\section{Note: concavity of hedge fund borrowing demand and MMF repo supply}
\footnote{\citep{Aronoffbilateralrepo}}
\label{app:Concavity}

This appendix provides microfoundations under which (i) hedge fund borrowing demand and (ii) money market fund (MMF) repo supply are concave in the relevant pricing margin. The primitives mirror the frictions documented empirically and used theoretically in the cited literature—risk aversion and funding/margin constraints for leveraged investors \citep{HeNagelSong2022JFE,BanegasMonin2023FEDS,GarleanuPedersen2011RFS} and concentration costs plus an outside option (the ON RRP floor) for MMFs \citep{Huber2023JFE,HempelMonin2023FEDS}.

\paragraph{Hedge funds (borrowing demand as a function of the repo rate).}
Let a hedge fund choose leverage $L\ge 0$ on an arbitrage (or carry) opportunity with expected excess return $\alpha$ over the secured financing rate, financed at repo rate $r$. The fund is risk‐averse with mean–variance preferences and faces convex funding/liquidity costs that increase with scale (capturing balance–sheet usage, margin frictions, and liquidity risk), consistent with \citep{HeNagelSong2022JFE,GarleanuPedersen2011RFS}. Consider the static problem
\[
\max_{S\ge 0}\;\; \underbrace{S(\alpha - r)}_{\text{gross carry}} \;-\; \underbrace{\tfrac12(\gamma\sigma^2 + k)\,S^2}_{\text{risk + quadratic funding cost}} \;-\; \underbrace{\tfrac{m}{3}\,S^3}_{\text{scale/liquidity convexity}},
\]
with $\gamma>0$, $\sigma^2>0$, $k\ge 0$, $m>0$. The FOC for interior $L$ is
\[
\alpha - r \;=\; (\gamma\sigma^2 + k)\,S \;+\; m\,S^2 \;\equiv\; A\, S + m L^2,\quad A:=\gamma\sigma^2+k.
\]
Solving for the Treasury suply $S(r)$ gives the positive root
\[
S(r)\;=\;\frac{-A + \sqrt{A^2 + 4m(\alpha - r)}}{2m}\qquad \text{for } r<\alpha,
\]
and $S(r)=0$ otherwise. Differentiating,
\[
\frac{dS}{dr}\;=\;-\frac{1}{A+2mS}\;<\;0,\qquad
\frac{d^2S}{dr^2}\;=\;-\frac{2m}{(A+2mS)^3}\;<\;0.
\]
Hence $S(r)$ is decreasing and concave in $r$ on $( -\infty,\alpha )$. This aligns with the empirical evidence that higher funding costs or margins sharply reduce leverage (e.g., a 200bp minimum haircut reducing effective leverage from $56\times$ to $25\times$ \citep{BanegasMonin2023FEDS}) and with margin–based asset–demand theories in which binding constraints flatten supply at the margin \citep{GarleanuPedersen2011RFS}. Moreover, when haircuts or margins introduce a binding constraint, $S(r)$ becomes piecewise defined with a kink (as in \citep{HeNagelSong2022JFE}), reinforcing curvature.

\paragraph{MMFs (demand for Treasuries from a given dealer as a function of the dealer’s offered rate).}
Fix a single dealer–MMF relationship (the per–dealer schedule relevant for dealer market power). Let the MMF choose exposure $D\ge 0$ to this dealer at offered rate $r$, with an outside option $r_{0}$ (e.g., the ON RRP floor) that creates a rate floor/kink \citep{HempelMonin2023FEDS}. Following the structural evidence on aversion to portfolio concentration and preference for stable lending, \citep{Huber2023JFE}, model concentration costs by a convex function $C(x)=\tfrac{a}{2}x^2+\tfrac{b}{3}x^3$ with $a>0$, $b>0$. The MMF solves
\[
\max_{D\ge 0}\;\; (r-r_{0})\,x \;-\; \frac{a}{2}x^2 \;-\; \frac{b}{3}x^3,
\]
implying the FOC (for $r\ge r_{0}$)
\[
r-r \;=\; aD + b D^2.
\]
The per–dealer supply function is the positive root
\[
D(r)\;=\;\frac{-a + \sqrt{a^2+4b(r-r_{0})}}{2b}\qquad \text{for } r\ge r,
\]
and $D(r)=0$ for $r<r$. By differentiation,
\[
\frac{dD}{dr}\;=\;\frac{1}{a+2b\,D}\;>\;0,\qquad
\frac{d^2D}{dr^2}\;=\;-\frac{2b}{(a+2b\,D)^3}\;<\;0.
\]
Thus, above the policy floor $r$ the MMF’s per–dealer demand is increasing and concave in the offered rate, with a kink at $r$ consistent with the ON RRP floor’s role in shaping the lending schedule \citep{HempelMonin2023FEDS}. Concavity captures the diminishing marginal willingness to concentrate exposure in a single counterparty documented in Huber \cite{Huber2023JFE}.

\paragraph{Discussion.}
Taken together, these primitives deliver concave Treasury supply on the leveraged–investor, $UB$, side (due to risk aversion and convex funding/liquidity costs) and concave per–dealer demand on the MMF, $UL$, side (due to concentration costs and an outside option). Both features are consistent with the empirical patterns reported in \citep{HeNagelSong2022JFE,BanegasMonin2023FEDS,Huber2023JFE,HempelMonin2023FEDS} and provide a parsimonious foundation for the curvature assumed in the main text.

\end{document}